\documentclass[a4paper, 11pt]{article}              
\RequirePackage[colorlinks,citecolor=blue,urlcolor=blue]{hyperref}     
\DeclareFontShape{OT1}{cmr}{bx}{sc}{<-> Combs10}{}   
\usepackage[oldstylenums]{kpfonts}   
\usepackage{amsmath, amssymb, amsthm} 
\usepackage{eufrak}
\usepackage{yfonts}
\usepackage{soul}  
\usepackage{microtype}
\usepackage{mathtools} 
\usepackage{cleveref}
\usepackage{enumerate} 
\usepackage{bm} 
\usepackage[utf8]{inputenc} 
\usepackage[T1]{fontenc}
\usepackage[normalem]{ulem}
\usepackage{epstopdf}
\crefname{equation}{}{}
\usepackage[]{geometry}
\usepackage{authblk}
\numberwithin{equation}{section}
\usepackage{dsfont}
\usepackage{graphicx} 
\usepackage{subcaption}
\usepackage[]{apacite} 
\usepackage[round, sort, compress, authoryear]{natbib} 
\theoremstyle{plain}

\newcommand{\FBA}{\text{FBA}}
\newcommand{\FCA}{\text{FCA}}
\newcommand{\FVA}{\text{FVA}}

\newcommand{\FBAd}{\text{FBA}^\Delta}
\newcommand{\FVAd}{\text{FVA}^\Delta}
\newcommand{\FCAd}{\text{FCA}^\Delta}

\newcommand{\dqdt}{\df \dsQ \otimes \df t-\text{a.s}}

\newcommand{\btau}{\bar{\tau}}

\newcommand{\cadlag}{c\`adl\`ag}

\newcommand{\1}{\mathds{1}}
\newcommand{\calF}{\mathcal{F}}

\newcommand{\calC}{\mathcal{C}}

\newcommand{\calB}{\mathcal{B}}

\newcommand{\frD}{\textfrak{D}}

\newcommand{\frC}{\mathfrak{C}}

\newcommand{\calG}{\mathcal{G}}

\newcommand{\calO}{\mathcal{O}}

\newcommand{\Bt}{\tilde{B}}
\newcommand{\St}{\tilde{S}}
\newcommand{\Sc}{\check{S}}
\newcommand{\Xc}{\check{X}}
\newcommand{\cc}{\check{c}}
\newcommand{\Vc}{\check{V}}
\newcommand{\pic}{\check{\pi}}

\newcommand{\pt}{\tilde{P}}
\newcommand{\pc}{\check{p}}

\newcommand{\pit}{\tilde{\pi}}

\newcommand{\dsI}{\mathbb{I}}  
\newcommand{\dsE}{\mathbb{E}}
\newcommand{\dsF}{\mathbb{F}}
\newcommand{\dsH}{\mathbb{H}}
\newcommand{\dsQ}{\mathbb{Q}} 
\newcommand{\dsP}{\mathbb{P}}

\newcommand{\dsL}{\mathbb{L}}
\newcommand{\dsR}{\mathbb{R}}
\newcommand{\dsS}{\mathbb{S}}
\newcommand{\dsG}{\mathbb{G}}
\newcommand{\dsD}{\mathbb{D}}

\newcommand{\Jt}{\tilde{\Theta}}

\newcommand{\SUM}{\displaystyle\sum}

\newcommand{\IT}{^i_t}

\newcommand{\IS}{^i_s}

\usepackage{bbding}
\usepackage{thmtools}
    
\theoremstyle{plain}
\newtheorem{thm}{Theorem}[section] 
\newtheorem{corollary}[thm]{Corollary} 
\newtheorem{theorem}[thm]{Theorem} 
\newtheorem{lemma}[thm]{Lemma}

\newtheorem{assumption}[thm]{Assumption}
\newtheorem{definition}[thm]{Definition} 

\theoremstyle{definition}
\newtheorem{remark}[thm]{Remark} 
\newtheorem{example}[thm]{Example}

\newcommand*\df{\mathop{}\!\mathrm{d}}

\newcommand{\Rom}[1]{\uppercase\expandafter{\romannumeral #1\relax}}
\newcommand{\rom}[1]{\lowercase\expandafter{\romannumeral #1\relax}}

\newcommand{\RN}[1]{%
  \textup{\uppercase\expandafter{\romannumeral#1}}%
}

\newcounter{current-cite}

\newcounter{currentcitetotal}

\renewenvironment{abstract}{\small
  \begin{center}
  \bfseries \abstractname\vspace{-.5em}\vspace{0pt}
  \end{center}
  \list{}{
    \setlength{\leftmargin}{0.7cm}%
    \setlength{\rightmargin}{\leftmargin}%
  }%
  \item\relax}   
 {\endlist}

 \title{Binary Funding Impacts in Derivative Valuation
   \thanks{\textbf{Acknowledgments: } The authors are grateful to the anonymous
     co-editor, associate editor, and reviewers for their invaluable efforts to
     improve this paper. The authors also thank St\'ephane Cr\'epey for
     insightful comments. Finally, the authors express gratitude to Dylan
     Possama\"{i} for discussions at the Workshop on Risk Measures, XVA
     Analysis, Capital Allocation and Central Counterparties (2016).}  }
 \author{Junbeom Lee\thanks{ Department of Sales and Trading, Yuanta Securities
     Korea, 04538 Seoul, Korea.  Email:
     \href{junbeoml22@gmail.com}{junbeoml22@gmail.com}.
     Most  of this work was carried out when the author was at the National
     University of Singapore. Opinions in this paper are those of the author,
     and do not represent the views of Yuanta Securities Korea.
     This author is supported by Singapore MOE AcRF grants R-146-000-255-114 and
     R-146-000-243-114.  } \quad \quad \quad
   Chao Zhou \thanks{
     Department of Mathematics, Institute of Operations Research and Analytics
     and Suzhou Research Institute, National Univesity of Singapore, Singapore
     119076, Singapore.
     Email: \href{mailto:matzc@nus.edu.sg}{matzc@nus.edu.sg}.
     This author is supported by Singapore MOE
     AcRF Grants  R-146-000-271-112, R-146-000-255-114 and R-146-000-219-112. }}

\date{}
\providecommand{\keywords}[1]{\textbf{\textit{Key-words: }} #1}
\providecommand{\AMS}[1]{\textbf{\textit{AMS subject classifications: }} #1}
\begin{document}
\begin{titlepage}
  \vspace*{-1cm} 
{\let\newpage\relax\maketitle}   
\begin{abstract} 
  We discuss the binary nature of funding impact in derivative valuation.  Under
  some conditions, funding is either a cost or a benefit, i.e., one of the
  lending/borrowing rates does not play a role in pricing derivatives.  When
  derivatives are priced, considering different lending/borrowing rates leads to
  semi-linear BSDEs and PDEs, and thus it is necessary to solve  the equations
  numerically. However, once it can be guaranteed that only one of the rates affects
  pricing, linear equations can be recovered and analytical formulae can be derived.
  Moreover, as a byproduct, our results explain how \textit{debt value
    adjustment} (DVA) and funding benefits are dissimilar.  It is often believed
  that considering both DVA and funding benefits results in a double-counting
  issue but it will be shown that the two components are affected by different
  mathematical structures of derivative transactions. We find that
  funding benefit is related to the decreasing property of the payoff function, 
  but this relationship decreases as the funding choices of underlying
  assets are transferred to repo markets.
\end{abstract}
\vspace{0.3cm} 
\noindent\keywords{BSDEs; Malliavin calculus; bilateral contracts; funding benefits; non-linear markets.}

\noindent\AMS{60H07, 91G20, 91G40}
\end{titlepage}

\addtocounter{page}{1}

\section{Introduction}
For bank's trading desks, it had been a standard to make \textit{credit value
  adjustment} (CVA) and \textit{debt value adjustment} (DVA). After the
financial crisis in 2007-2008, banks started to consider their funding spreads
arising from the widened gap between the London Inter-Bank Offered Rate (LIBOR)
and the Overnight Indexed Swap (OIS) rate. Since their lending and borrowing
rates differ significantly, non-linear financial market models are requisite for
pricing derivatives with funding impacts. These general practices forced us to
re-examine existing theory. Definitions of fair value under entity-specific
parameters such as default risk and funding spreads are introduced by
\cite{bielecki2015valuation, bielecki2018arbitrage, bichuch2017arbitrage,
  nie2018american}. In addition, many other pricing methodologies have been
developed for collective adjustment, see, for example,
\cite{piterbarg2010funding, wu2015cva, li2016fva, brigo2011collateral,
  crepey2015bilateral1, crepey2015bilateral2}.

However, certain conundrums among the adjustments have remained. One of these
conundrums is whether DVA double-counts funding benefits. Funding benefits
reduce the amount of funds that the bank needs to raise, but the cost of raising
funds depends on the bank's default risk.  Therefore, including both DVA and
funding benefits may inflate the bank’s reported profit as well as deflate the
price charged to counterparties, see, e.g., \cite{cameron2013black} and Chapter
15 in \cite{gregory2015xva}.  Moreover,  considering banks' own funding rates can be
arguable. According to the Modigliani-Miller (MM) theorem, choices of funding should
not be considered in pricing \citep{modigliani1958cost, stiglitz1969re}.  In
practice, traders feel confident that funding costs are observable in derivative
transactions \citep[see][]{andersen2019funding}).  As pointed out by
\cite{hull2012fva}, however, if funding rates are truly an element to determine
derivative prices, the existence of Treasury bonds is enigmatic because banks
purchase Treasury bonds that return less than their funding rates.

The above issues will be discussed with the main results of this paper. Our main
contribution is demonstrating the binary nature of funding impacts, in which
funding is either a benefit or a cost for many derivative contracts. For
example, it will be shown later that, based on our model, when purchasing bonds,
the trader will never enter a borrowing position, i.e., there is no funding cost
and funding is always a benefit in the case. Consequently, if we assume that the
lending rate is equal to the OIS rate \citep[for example, as
in][]{burgard2010partial}, no funding impact should be considered. This may
explain why banks do not recoup funding costs when purchasing Treasury bonds,
while funding costs are observable in other derivative transactions. In
addition, the switching between funding costs and benefits depends on the
monotone property of the payoff function of derivatives. Since DVA occurs where
the seller has an obligation to pay, i.e., where the payoffs are positive, DVA
and funding benefit are affected by different mathematical structures of the
payoffs. This result supports a part of the conclusion in
\cite{andersen2019funding} that DVA should be considered in pricing.

Another contribution of our results is that we can recover linear equations to
price derivatives. Because different lending/borrowing rates make the pricing
equations semi-linear, analytical solutions do not exist in general. Therefore,
we need to solve the equations numerically. For an attempt to approximate
funding adjustment of contracts with short maturity, readers may want to refer
to \cite{gobet2015analytical}. Moreover, under the crucial assumption that
lending and borrowing rates are identical, closed-form solutions were calculated
by \citep{piterbarg2010funding, bichuch2017arbitrage, brigo2017funding}.
However, once we can guarantee that one of funding rates does not play any role
in pricing derivatives, the assumption of the same borrowing/lending rates is
superfluous.

Even though this binary nature enables us to find analytical formulae for a
large class of derivatives, depending on the close-out conventions, the
analytical formulae may or may not be represented in a closed-form. In this
paper, we deal with two close-out conventions: \textit{clean price} and
\textit{replacement cost}. \textit{Clean price} (resp. \textit{replacement
  cost}) is the risk-neutral price without (resp. with) value adjustments. Under
\textit{clean price}, because pricing measures are not matched, the analytical
formulae cannot be represented in a closed-form. Indeed, to avoid this
inconsistency, \cite{bichuch2017arbitrage} assumed a flexibility to choose a
pricing measure for calculating close-out amount, and \cite{brigo2017funding}
considered only un-collateralized contracts with null cash-flow at defaults. In
the case of \textit{replacement cost}, the mismatch does not appear, because, in
our model, the funding rates are tacitly embedded in the \textit{replacement
  cost}.  In this way, we can provide closed-form solutions under
\textit{replacement cost}.

In our model, we include CVA, DVA, incremental funding impacts, and 
\textit{variation margin}. The reference filtration is generated by a Brownian
motion. We then progressively enlarge the filtration by default times of the two
parties. The default times are assumed to have intensities. Since we do not assume
that interest rates are deterministic, our results can be applied to interest
rate derivatives. In the main theorems, we assume that volatility and default
intensities are deterministic to avoid heavy calculations, but stochastic
parameters do not necessarily alter the main results.  We report one example
in which intensities are not deterministic in the Appendix.

The remainder of this paper proceeds as follows. In \Cref{sec:modeling}, we
introduce our setup concerning filtration and intensities, and construct the
\textit{incremental hedging portfolio}. In \Cref{sec:bsde}, we introduce a BSDE
to price derivatives on the enlarged filtration. Instead of dealing with the
BSDE, we define an adjustment process in \Cref{sec:xva} and the adjustment
process is reduced to a BSDE on the reference filtration as in
\cite{crepey2015bsdes}. Our main results are then provided in
\Cref{sec:main}. To prove the main theorems, iterative transformations of the
adjustment process are needed and the transformations depend on the close-out
conventions. In \Cref{sec:example}, examples are given and we provide a
closed-form solution for a stock call option with \textit{replacement cost}.
 
\section{Modeling}\label{sec:modeling}

\subsection{Mathematical Setup}\label{sec:setup}
We consider two parties entering a bilaterally cleared contract. We term one
party a ``hedger'' and the other party a ``counterparty''. We sometimes address
the hedger (resp. counterparty) as ``she'' (resp. ``he''). An index $H$
(resp. $C$) will be used to represent the hedger (resp. counterparty). The
argument of this paper is conducted from the perspective of the hedger. The
hedger is a financial firm that holds a portfolio to hedge the exchanged
cash-flows of the contract. The counterparty may or may not be a non-financial
firm.  Note that when we refer to  a ``dealer'' or ``bank'', we are not necessarily
referring to the hedger since the counterparty can also be a bank.

Let $\dsQ$ be a risk-neutral probability where the account of OIS rate is the
num\'eraire. We will impose an assumption later so that the risk-neutral
probability is unique. Then we consider $(\Omega, \calG, \dsQ)$.  Let $\dsE$ denote
the expectation under $\dsQ$. We consider random times $\tau^i$,
$i \in \{H, C\}$, $\tau^i \colon (\Omega, \calG) \to (\dsR_+, \calB(\dsR_+))$, which
represent the default times of the hedger and counterparty. For $i\in \{H, C\}$,
we assume that $ \dsQ(\tau^i=0) =0$ and
$\dsQ(\tau^i>t) > 0, ~ t \geq0$. We also denote
 $\tau \coloneqq  \tau^H \wedge \tau^C$, $\btau \coloneqq \tau \wedge T$,
where $T$ is the maturity of the bilateral contract.

Let $W=(W^1, \dots, W^n)$ be a standard $n$-dimensional Brownian motion under
$\dsQ$. Let $\dsF = (\calF_t)_{t \geq0}$ be the \textit{usual natural filtration} 
of $(W_t)_{t \geq0}$. Then we define the following:
\begin{align}
 \dsG = (\calG_t)_{t \geq0}
  \coloneqq \Big(\calF_t \vee \sigma\big(\{\tau^i \leq u\}
  ~\colon~ u \leq t, i \in \{H, C\} \big)\Big)_{t \geq0}.  \nonumber 
\end{align}
We term $\dsF$ (resp. $\dsG$) the reference filtration (resp. full filtration). 

We then consider a filtered probability space $(\Omega, \calG, \dsG,\dsQ)$.  Note
that $\tau^i$, $i \in \{H, C\}$, are $\dsG$-stopping times, but may not be
$\dsF$-stopping times.  As a convention, for
any $\dsG$-progressively measurable process $u$, $(\dsG, \dsQ)$-semimartingale
$U$, and $s \leq t$,
$\int_{s}^{t}u_s \df U_s \coloneqq \int_{(s, t]}^{}u_s\df U_s$, where the integral is
well-defined. In addition, for any $\dsG$-stopping time $\theta$ and process
$(\xi_t)_{t \geq0}$, we denote
$\xi^\theta_\cdot \coloneqq \xi_{\cdot \wedge \theta}$, and when $\xi_{\theta-}$ exists,
$\Delta \xi_\theta\coloneqq \xi_\theta - \xi_{\theta-}$.  In the following, for
$i \in \{H, C\}$, $t \geq0$, we denote
$G\IT \coloneqq \dsQ(\tau^i > t|\calF_t)$, and
$G_t\coloneqq \dsQ(\tau > t|\calF_t)$.  The following assumption stands throughout
this paper.
\begin{assumption}\label{assm:intensity}
\begin{enumerate}[(i)]
\item $(G_t)_{t \geq0}$ is non-increasing and absolutely continuous
  with respect to Lebesgue measure. 
\item For any $i \in \{H, C\}$, there exists a process $h^i$, defined as:
\begin{align}
  h^i_t \coloneqq \lim_{u\downarrow0}\frac{1}{u}\frac{\dsQ(t < \tau^i \leq t+u, \tau > t|\calF_t)}{\dsQ(\tau >
  t|\calF_t)}, \nonumber 
\end{align}
and the process $M^i$, given by
$M^i_t \coloneqq \1_{\tau^i \leq t \wedge \tau}   - \int_{0}^{t \wedge \tau}h^i_s \df s$,
is a $(\dsG, \dsQ)$-martingale. 
\end{enumerate} 
\end{assumption}  
By (\rom{1}) in
\Cref{assm:intensity}, there exists an $\dsF$-progressively measurable process
$(h^0_t)_{t\geq0}$ such that:
\begin{align}
   h^0_t = \lim_{u\downarrow0}\frac{1}{u}\frac{\dsQ(t < \tau \leq t+u|\calF_t)}{\dsQ(\tau >
  t|\calF_t)},  \nonumber 
\end{align}
and
$(M_t)_{t\geq0}\coloneqq(\1_{\tau\leq t} - \int_{0}^{t\wedge\tau}h^0_s\df s)_{t\geq0}$ 
is also a $(\dsG, \dsQ)$-martingale. Let us denote $h \coloneqq h^H + h^C$. If
$\tau^H$ and $\tau^C$ are independent  conditioning on $\dsF$, $h^0 = h$. This is, in general, not the
case.  Moreover, by (\rom{1}) in \Cref{assm:intensity}, $\tau$ avoids any
$\dsF$-stopping time \cite[see][Corollary 3.4]{coculescu2012hazard}. In other
words, for any $\dsF$-stopping time $\tau^\dsF$:
\begin{align} 
 \dsQ(\tau = \tau^\dsF) = 0.  \label{avoid}
\end{align}
We use the typical notations $\dsL^p_T$, $\dsS^p_T$, $\dsH^{p, m}_T$, and
$\dsH^{p, m}_{T, loc}$ for spaces of random variables and stochastic processes, which
are defined in \Cref{sec:spaces}. 
Moreover, we let $D_{\theta}=(D^1_\theta, \dots, D^n_\theta)$ denote Malliavin derivative at
$\theta \geq0$, and $\dsD^{1, 2}$ denote the set of Malliavin differentiable random
variables. For Malliavin calculus, readers can refer to \cite{di2009malliavin}
and Section 5.2 in \cite{el1997backward}. 
For notational simplicity, when $n=1$, we denote $D_\theta = D^1_\theta$,
$W = W^1$, and $\dsH^p_T\coloneqq \dsH^{p, 1}_T$,
$\dsH^p_{T, loc}\coloneqq \dsH^{p, 1}_{T, loc}$.
In the next section, we describe
hedging portfolios under CVA, DVA, funding impacts, and collateral.
\subsection{BSDEs under Nonlinear Markets}
\subsubsection{Cash-flows}
We consider a hedger and a counterparty entering a contract which exchanges
promised dividends. Let $\frD_t$ denote the accumulated amount of the promised
dividends up to $t \geq0$. We assume that $\frD$ is an $\dsF$-adapted c\`ad\`ag
process, and the value is determined by an $n$-dimensional underlying asset
process $S=(S^1, \dots, S^n)$ that follows the  stochastic differential
equation (SDE) under $\dsQ$:
\begin{align}
  \df S\IT =& r_tS\IT \df t + \sigma\IT S\IT \df W^{}_t, ~~S_0^i >0, 
~~i \in\{ 1, \dots, n\},\label{non-defaultable}
\end{align}
for some $\dsF$-progressively measurable processes $(\sigma^i)^\top\in \dsR^n$ and
$r$, which represents the OIS rate. We assume the following:
\begin{assumption}
 The hedger can access  defaultable zero coupon bonds of the 
hedger and the counterparty.  
\end{assumption}
Let $S^H$ (resp. $S^C$) denote the defaultable bond of
the hedger  (resp. the counterparty), where $S^H$ and $S^C$ follow the next SDE:
\begin{align}
  \df S\IT =& r_tS\IT \df t + \sigma\IT S\IT \df W^{}_t-S^i_{t-}\df M^i_t,~~S_0^i>0,
              ~~i \in\{ H, C\},\label{defaultable}
\end{align} 
where $(\sigma^i)^{\top} \in \dsR^{n}$ are $\dsF$-progressively measurable
processes. $S^1, \dots, S^n$ are used to hedge $\frD$, while $S^H$ and $S^C$ are
used to hedge the loss from breaching of the contract,
which will later be denoted by $\Theta$.  We denote
$\Sigma\coloneqq [(\sigma^1)^\top \cdots (\sigma^n)^\top]^\top$ and
$\sigma \coloneqq [\Sigma^\top~ (\sigma^H)^\top~ (\sigma^C)^\top]^{\top}$. In the following, we assume:
\begin{assumption}
$\Sigma\in \dsR^{n\times n}$ is invertible.  
\end{assumption}
By this assumption, we justify the existence of a risk-neutral probability
measure which is given in the early stage of this paper in
\Cref{sec:setup}. Moreover, the risk-neutral probability is unique. A brief
explanation for the change of measure is provided in \Cref{sec:risk-neutral}.
\begin{remark}
  By (\ref{non-defaultable}), underlying assets do not depend on default risk,
  which means that we do not deal with credit derivatives. For modeling with
  emphasis on contagion risk, readers can refer to
  \cite{jiao2013optimal, brigo2014arbitrage, bo2019locally, bo2017credit}. 
\end{remark}

If a default occurs prior to the maturity of the contract $T$, two parties stop
exchanging $\frD$, and the derivative contract is marked-to-market. The method
to calculate the close-out amount is determined before initiation of the
contract and documented in the Credit Support Annex (CSA)\footnote{A part of
  ISDA master agreement}. Let $e_t$ denote the close-out amount at $t \leq T$.  In
this paper, we deal with two conventions for the close-out amount $e$:
\textit{clean price} and \textit{replacement cost}.  We will explain the
conventions in the subsequent section after we define the hedger's hedging
portfolio. As conventions, $\df \frD_t \geq0$ (resp. $<0$) and $e_{t} \geq0$
(resp. $<0$) mean that the hedger pays to (resp. receives from ) the
counterparty at $t \leq T$. For example, if the hedger purchases a zero coupon bond of
unit notional amount, then $ \frD = -\1_{\llbracket T, \infty\llbracket}$.
 
The obligation to settle $e$ may not be fully honored due to the default. To
mitigate this risk, the two parties post or receive collateral (often referred to
as margin). The amount of the collateral posted at $t\geq0$ is denoted by $m_t$.
We assume that $(m_t)_{t\geq0}$ is an $\dsF$-adapted process. By
\Cref{assm:intensity}, $\tau^i$, $i \in \{H, C\}$, are totally inaccessible, which
means that the defaults arrive unexpectedly. Margins are posted because we do
not know full information about the defaults, and this is why $m$ is calculated
on the observable information $\dsF$.  The exact forms of $m$ will be given
later after the conventions for $e$ are introduced. We assume that the
close-out payment is settled at the moment of default without delay and $m$ is
posted continuously.  As a convention, if $m_t \geq0$ (resp.  $<0$),    
the hedger posts (receives) the collateral at $t \leq T$.

\begin{remark}
  In practice, a gap exits between the day of default and the actual settlement
  to verify whether the default really happened, collect information about the
  contract, and find the best candidate to replace the defaulting party
  \citep{murphy2013otc}. For \textit{gap risk}, two parties post an
  \textit{initial margin}. If we consider \textit{initial margin}, we encounter
  anticipative backward stochastic differential equations (ABSDEs) under
  \textit{replacement cost}.  For the main result of this paper, Malliavin
  calculus for BSDEs will be utilized. However, to the best of our knowledge,
  Malliavin differentiability of ABSDEs has not yet been investigated. The
  continuous posting of \textit{variation margin} can also be viewed as a
  simplification. For relaxation of the condition, readers may refer
  to  \Citet{brigo2014nonlinear}. 
\end{remark}       
At default, collateral is not exchanged. We set the collateral amount at
$\tau \leq T$, as $m_{\tau-}$, and thus the cash-flow at default can be
$\Delta\frD_{\tau} + e_{\tau} - m_{\tau-}$.  However, whether or not we separate
$\Delta\frD _{\tau}$ from $e $ in the modeling is immaterial, because jumps of
$\dsF$-adapted \cadlag\ processes are exhausted by $\dsF$-stopping times
\citep[see][Theorem 4.21]{he1992semimartingale}. Therefore, by (\ref{avoid}),
$\Delta \frD _{\tau} =0$, a.s.  Let $\frC $ denote the accumulated cash-flows. Then,
almost surely for any $t \leq T$:
\begin{align}
\frC _t \coloneqq \1_{\tau > t}\frD _t + \1_{\tau \leq t}(\frD _{\tau} + e _{\tau})
  -\1_{\tau =\tau^H\leq t}L^H(e _{\tau} - m _{\tau-})^+
  + \1_{\tau=\tau^C\leq t}L^C(e _{\tau}   - m _{\tau-})^-,\label{def.cf} 
\end{align}
where $0\leq L^H\leq1$ (resp. $0\leq L^C\leq1$ ) is the loss rate of the hedger
(resp. counterparty). We denote the following:
\begin{align}
   \Theta\coloneqq\frC -\frD^\tau. \label{def:J}
\end{align}
In other words,
$\Theta_t = \1_{\tau = \tau^H \leq t}\big[e_\tau - L^H(e _{\tau} - m _{\tau-})^+\big] +\1_{\tau = \tau^C \leq
  t}\big[e_\tau + L^C(e _{\tau} - m _{\tau-})^-\big]$.  Generally, $\Theta-e$ represents the
loss inflicted to the hedger from breaching of the contract.
\subsubsection{Accounts and Hedging Strategy}
In this section, we introduce several saving accounts and the hedger's
hedging strategy. Let $I \coloneqq \{1, 2, \dots, n, H, C\}$, and
for any $i \in I$, $\eta^{S, i}$ denote the number of units of $S^i$ held by the
hedger. We assume that $\eta^{S, i}$ is $\dsG$-predictable. We denote:
\begin{align}
  \varphi \coloneqq& (\eta^{S, 1}, \dots, \eta^{S, n, }, \eta^{S, H}, \eta^{S, C}),
                     \nonumber \\
  \pi^i \coloneqq& \eta^{S, i}S^i, ~~i \in \{1, \dots, n\} ,\nonumber \\
\pi^i \coloneqq& \eta^{S, i}S^i_-, ~~i \in \{H, C\} ,\nonumber \\
\pi \coloneqq& (\pi^1, \dots, \pi^n, \pi^H, \pi^C) .\nonumber 
\end{align}
We call the $(n+2)$-dimensional $\dsG$-predictable process $\varphi$ the
hedger's hedging strategy. 
\begin{remark}
  $\varphi$ is chosen to be $\dsG$-adapted only to describe an immediate action
  taken at default. It will be shown in \Cref{sec:bsde} that, on $[0, \tau)$,
  $\varphi$ is $\dsF$-adapted.
\end{remark}

If the collateral is pledged, the posting party is remunerated by the receiving
party according to a certain interest rate. When $m_t\geq0$ (resp. $<0$), the
counterparty (resp. hedger) pays the interest rate $R^{m, \ell}_t$ (resp.
$R^{m, b}_t$)  at $t \leq T$. We assume  that the collateral is posted
as cash and the interest rate is accrued to a margin account of the hedger. We 
denote the lending and borrowing accounts $B^{m, \ell}$ and $B^{m, b}$
respectively, i.e., $B^{m, i}$, $i \in\{\ell, b\}$,  are given by:
\begin{align}
 \df B^{m, i}_t = R^{m, i}_tB^{m, i}_t\df t, ~~B^{m, i}_0=1.\label{sde:margin}
\end{align}
Let $\eta^{m, \ell}$ (resp. $\eta^{m, b}$) indicate the number of units of
$B^{m, \ell}$ (resp. $B^{m, b}$). Then the following equations hold:
\begin{eqnarray}
  &\eta^{m, \ell} \geq0, ~~\eta^{m, b}\leq0, ~~\eta^{m, \ell}\eta^{m, b}=0, \label{cond:margin1}\\
 & \eta^{m, \ell}B^{m, \ell} + \eta^{m, b}B^{m, b} = m. \label{cond:margin2}
\end{eqnarray}
We assume that the  \textit{variation margin} $m$ can be
\textit{rehypothecated}, i.e., $m$ is used by the hedger to maintain her
portfolio.

Some underlying assets can be traded through repo markets. We denote the set of
indices  for which a repo market is available by
$ \rho \subseteq I \coloneqq \{1, 2, \dots, n, H, C\}$. We assume that the borrowing and
lending repo market rates are the same, and for $i \in \rho$, let $R^{\rho, i}$ indicate
the repo rate. Moreover, for any $i \in\rho$, let $B^{\rho, i}$ denote the account that
$R^{\rho, i}$ accrues, i.e., $B^{\rho, i}$ follows:
\begin{align}
\df B^{\rho, i}_t = R^{\rho, i}_tB^{\rho, i}_t\df t, ~~B^{\rho, i}_0=1. \label{sde:repo}
\end{align}
For $i \in \rho$, we denote  the number of units of $B^{\rho, i}$ by $\eta^{\rho, i}$. It then
follows that for any $i \in \rho$:
\begin{align}
 \eta^{\rho, i}B^{\rho, i} + \eta^{S, i}S^{i} = 0.\label{cond:repo}
\end{align}

If the hedger has any surplus cash, she can earn the  lending rate $R^\ell$;
whereas, when
borrowing money, she needs to pay the borrowing rate $R^b$. For $i \in
\{\ell, b\}$, let $B^i$ denote the hedger's funding account and $\eta^i$ denote the
number of units of $B^i$. Therefore, it follows that:
\begin{eqnarray}
  & \eta^\ell \geq0, ~~\eta ^b \leq 0, \\
  &\df B^i_t = R^i_tB^i_t\df t, ~~B^i_0=1, ~~i\in \{\ell, b\}. \label{sde:funding}
\end{eqnarray}
Finally, the account on the OIS rate is denoted by $B$, i.e., for $t\geq0$, $\df B_t =
r_t B_t \df t$ and $B_0=1$. 
\subsection{Hedger's Hedging Portfolio}
We are now ready to define the hedger's portfolio. 
\begin{definition} \label{def:iself}
If $V=V(V_0, \varphi, \frC)$ defined on $t \in \dsR_+$, by:
\begin{align}
  V_t = \eta^{\ell}_tB^\ell_t +\eta^{b}_tB^b_t+ \eta^{m,\ell}_tB^{m, \ell}_t+\eta^{m, b}_tB^{m, b}_t
  +\SUM_{i \in I}\eta^{S, i}_t  S\IT+\SUM_{i \in \rho}\eta^{\rho, i}_t B^{\rho, i}_t, \label{self}
\end{align}
satisfies:
\begin{align}
  V_t =& V_0 +\SUM_{i =\ell, b}\int_{0}^{t \wedge \btau}\eta^{i}_s\df B^i_s
         + \SUM_{i =\ell, b}\int_{0}^{t \wedge \btau}\eta^{m,i}_s\df B^{m, i}_s +\SUM_{i \in
         I}\int_{0}^{t \wedge \btau}\eta^{S, i}_s  \df S\IS \nonumber\\
       &+\SUM_{i \in \rho}\int_{0}^{t \wedge \btau}\eta^{\rho, i}_s \df B^{\rho, i}_s
         - \frC_{t\wedge \btau},
    \label{selfd}
\end{align}
for any $t \in \dsR_+$, then $V$ is called the hedger's hedging
  portfolio. 
\end{definition}
\begin{remark}
 Note that by (\ref{def.cf}),  $\frC_t = \frC_{t \wedge\btau}$, for any $ t \geq0$, and by
 (\ref{selfd}), $V_t = V_{t \wedge \btau}$, for any $ t \geq0$. 
\end{remark}

Our goal is to find a fair price charged to the counterparty, and a hedging
strategy $\varphi$ satisfying \Cref{def:iself}. We seek to impose a terminal
condition so that an incremental funding effect can be considered. This
incremental funding effect means the difference between the funding cost/benefit
of two choices: entering or not entering the new contract. To explain the
mathematical details, let $B^\epsilon$ denote the endowed bank's portfolio
without entering the new contract, for some $\epsilon \in \dsR$ such that:
$B^\epsilon_{0} = \epsilon$. We term $B^\epsilon$ a \textit{legacy
  portfolio}. We assume that the \textit{legacy portfolio} is locally risk-less
and grows with respect to its funding rates. Therefore, for any $t \geq0$:
\begin{align}
  B^\epsilon_t \coloneqq & \epsilon\exp{\Big(\int_{0}^{t}R^\epsilon_s \df s\Big)},  \nonumber\\
  R^\epsilon \coloneqq &\1_{\epsilon \geq0} R^\ell +   \1_{\epsilon <0} R^b,\nonumber
\end{align}
and we denote $s^\epsilon\coloneqq R^\epsilon -r$.

Now, let us regard $V$ as the bank's summed profit/loss and consider two
cases. First, if the bank does not enter the contract exchanging $\frC$, the
bank will have $B^\epsilon_{\btau}$ at $\btau$. Second, the bank can enter the contract
with a certain initial price,  denoted by $p\in \dsR$, for the
contract from the counterparty. We will investigate a fair value $p \in \dsR$
taking the opportunity cost of not entering the contract into account.
\begin{definition}
  Let $p\in \dsR$ and $\varphi$ be an $\dsR^{n+2}$-valued $\dsG$-predictable
  process. We call $p$ and $\varphi$ the \textit{replicating price} and
  \textit{replicating strategy} of $(\epsilon, m, \frC)$, respectively, if
  $V_{\btau}(p+\epsilon, \varphi, \frC) = B^\epsilon_{\btau}$.
\end{definition}
By (\ref{selfd}) together with (\ref{non-defaultable}), (\ref{defaultable}),
(\ref{sde:margin}), (\ref{cond:margin1}), (\ref{cond:margin2}),
(\ref{sde:repo}), (\ref{cond:repo}), (\ref{sde:funding}), it is easy
to confirm that the replicating price can be obtained by:
\begin{align}
p=V_0 - \epsilon    \label{i.price}
\end{align}
where $V$ is the solution of the following BSDE (under $\dsQ$):
\begin{align}\label{selfd2} 
V_t = 
  &B^\epsilon_{\btau}+\int_t^{\btau}\df \frC_s
    -\int_{t}^{\btau}\bigg[\big(V_s - m_s - \SUM_{i \in I \setminus \rho} \pi^i_s\big)^{+}R^{\ell}_s
      -\big(V_s - m_s - \SUM_{i \in I \setminus \rho}\pi^i_s\big)^{-}R^{b}_s
      +r_s\SUM_{i \in I}\pi\IS-\SUM_{i \in \rho}R^{\rho, i}_s\pi\IS\bigg]\df s\nonumber\\
  &-\int_{t}^{\btau}\bigg[ R^{m, \ell}_sm^+_s - R^{m, b}_sm^+_s\bigg]\df s
    -\SUM_{i \in I}\int_{t}^{\btau}\pi\IS\sigma\IS\df W_s
    + \SUM_{i =H, C}\int_{t}^{\btau}\pi^i_s\df M^i_s.
\end{align} 
For now, we do not examine the existence and uniqueness of
(\ref{selfd2}). The solvability will be examined with a reduced form of
(\ref{selfd2}). The financial interpretation of each component in (\ref{selfd2})
will be provided  in the following section.
Prior to continuing, readers may want to refer to \Cref{sec:inc}, where we
discuss the incremental funding impacts with a simple example.

In the following, for simplicity, we impose some realistic assumptions on
interest rates. In practice, $R^{m, i}$, $i \in \{\ell, b\}$, are chosen as Federal
funds or EONIA rates, i.e., approximately $r$. In addition, the difference between the
OIS
and repo market rates can be interpreted as the liquidity premium of the repo
markets.  We assume that the repo markets are sufficiently liquid for the difference to be
small. Moreover, we assume that the OIS rate, $(r_t)_{t\geq0}$, is the smallest among
all interest rates. 
These assumptions are summarized as follows:
\begin{assumption}\label{assm:rate.exo}
\begin{enumerate}[(i)]
\item $R^{\rho, i} = R^{m, \ell} =R^{m, b}= r$, for $i \in \rho$.
\item $R^\ell \geq r$ and $R^b\geq r$.
\end{enumerate}
\end{assumption}
\begin{remark}
The assumption on repo market rates given in \Cref{assm:rate.exo}-(\rom{1}) is  
merely for simplicity in representing (\ref{selfd2}). Mathematically, it does
not play a  crucial role.
\end{remark}
Recall that we have not yet specified the close-out amount $e$ in
$\frC$. In the next section, two important close-out conventions are introduced:
\textit{clean price} and \textit{replacement cost}. \textit{Clean price} will
also play an important role in reducing filtration of the BSDE \cref{selfd2}.
\subsubsection{Close-out Conventions}
%

We consider two close-out conventions: \textit{clean price} and
\textit{replacement cost}. \textit{Clean price} is the classical risk-neutral
price without any adjustment. Let $P$ denote the \textit{clean price}:
\begin{align}
  P_t \coloneqq B_t\dsE\bigg(\int_{t}^{T}B_s^{-1}\df \frD_s\bigg\vert\calF_t\bigg),~~
   \text{for }t\leq T. \label{def:cleanprice}
\end{align}
\textit{Clean price}  has often been chosen
in the literature (for example,  \cite{crepey2015bilateral1,
  crepey2015bilateral2}).

By \textit{replacement cost}, we mean the price with CVA, DVA, funding rates,
and collateral. In this case, since it is unclear which funding rates should be
chosen, we assume that the replacing party has similar credit spreads to the
hedger. Furthermore, recall that $V-B^\epsilon$ is the value for calculating the
derivative price from the perspective of the hedger. Therefore, under \textit{replacement
  cost}, we set the following: 
\begin{align}
  e_{\tau}=V_{\tau-}-B^\epsilon_{\tau}. \nonumber
\end{align}
In both conventions, we assume that the collateral is proportional to
the close-out amount. More precisely, for $0\leq L^m \leq1$ (margin loss):
\begin{align}
m =(1-L^m)e. \label{form:margin}
\end{align}
Note that (\ref{form:margin}) is consistent with our financial modeling. Indeed,
if there is a $\dsG$-adapted process satisfying (\ref{selfd2}), in our 
filtration setup $\dsG$, there is an $\dsF$-adapted process $V^\dsF$, such that
  $V = \1_{\llbracket 0, \btau \llbracket}V^\dsF$. 
Therefore, under \textit{replacement cost}, the margin process:
\begin{align}
  m =& (1-L^m)e = (1-L^m)(V_- -B^\epsilon) = (1-L^m)(V^\dsF_- -B^\epsilon) \nonumber 
\end{align}
is $\dsF$-adapted before $\btau$. 
\begin{remark}
\begin{enumerate}[(i)]
\item In practice, the two close-out conventions have advantages and
  disadvantages in financial modeling.  Readers can refer to
  \cite{brigo2011close} for a comparison.
\item A similar collateral convention was discussed by
  \cite{burgard2010partial}. For BSDEs' approach to general endogenous
  collateral, readers can refer to \cite{nie2016bsde}.    
\end{enumerate}
  
\end{remark}
Prior to further argument, we provide a lemma on properties of \textit{clean
  price} $P$. The following lemma will be utilized to present an adjustment
process and simplify the representation of the amount of cash-flow at default
$\Theta$. Readers may want to recall (\ref{def:cleanprice}), the definition of
$P$, before the following lemma. The proof is reported in \Cref{app:lemmas}.
\begin{lemma}\label{lem:pf}
\begin{enumerate}[(i)]
\item $P_T = 0$.
\item $P_{\btau} = \1_{\tau \leq T}P_\tau$.
\item   $\df P_t = r_tP_t \df t +B_t(Z^P_t)^{\top}\df W_t - \df
  \frD_t$,  for $t \leq T$,  for some $Z^P \in \dsH^{2, n}_{T,  loc}$. 
\item $P_{\tau-} = P_{\tau}$ almost surely.
\end{enumerate}
\end{lemma}

\subsubsection{Incremental Adjustment Process}
\label{sec:xva}
We can remove $\frD$ in (\ref{selfd2}) by using (\rom{3}) in \Cref{lem:pf}. To
this end, we introduce an incremental adjustment process. In both close-out
conventions, we will deal with the adjustment process instead of (\ref{selfd2}).
The adjustment process is defined by the discounted difference between
$V-B^\epsilon$ and $P$. Let $X$ denote the (incremental) adjustment process:
\begin{align}
X \coloneqq B^{-1}[V-B^\epsilon -P^{\btau}]. \nonumber
\end{align}
Moreover, let $\pt \coloneqq B^{-1}P$, $\pit \coloneqq B^{-1}\pi$,
$c \coloneqq B^{-1}m$, $\Jt \coloneqq B^{-1}\Theta$ ($\Theta$ is defined in \cref{def:J}),
$\Bt^\epsilon\coloneqq B^{-1}B^\epsilon$, and
$s^i \coloneqq R^i - r, ~~i \in \{\ell, b\}$.  Note that $s^\ell$ (resp. $s^b$)
represents the lending (resp. borrowing) spread of the hedger. We can easily
verify that for $t \geq0$,
$\df \pt_{t\wedge \btau} = \1_{t \leq \btau}\df \pt_t$.  Assuming that there exists
$(V, \pi)$ satisfying (\ref{selfd2}), by applying It\^o's formula to $X$, we have
the following for $t \leq \btau$, $(X, \pit)$:
\begin{align}      \label{dyn-xva}
  \left\lbrace
  \begin{array}{r@{}l}
    \df X_t
  =& \bigg[\big(X_t + \pt_t +\Bt^\epsilon_s- c_t 
- \SUM_{i \in I \setminus \rho} \tilde{\pi}^i_t\big)^{+}s^{\ell}_t
     -\big(X_t + \pt_t +\Bt^\epsilon_s- c_t - \SUM_{i \in I \setminus \rho}\tilde{\pi}^i_t\big)^{-}s^{b}_t
     -s^\epsilon_t\Bt^\epsilon_t\bigg]\df t      \\ \vspace{0.2cm}
   &+\big[\tilde{\pi}^{\top}_t\sigma_t-(Z^P_t)^{\top}\big]\df W_t 
     - \tilde{\pi}^H_{t}\df M^H_t- \tilde{\pi}^C_{t}\df M^C_t, \\
 X_{\btau} =& \Jt_{\btau} - \pt_{\btau}.
  \end{array}
  \right.
\end{align}
 Note that under \textit{replacement cost}, $\Jt$ depends on $X$. More
 precisely: 
\begin{align}
\Jt_t = \Jt_t(X_-) =  \1_{\tau \leq t} \pt_{\tau-}+\1_{\tau=\tau^H\leq t}\big[X_{\tau-}-L^HL^m(X_{\tau-} +
  \pt_{\tau-} )^+\big] +\1_{\tau=\tau^C \leq t}\big[X_{\tau-}+L^CL^m(X_{\tau-} + \pt_{\tau-})^-
     \big],\nonumber 
\end{align}
while under \textit{clean price}, $\Jt$ is independent of $X$. 
We also define $\Theta^H(X_-)$ and  $\Theta^C(X_-)$ such  that: 
\begin{align}
  \Jt_{\btau} - \pt_{\btau} =
  -\1_{\tau =\tau^H \leq T}\Jt^H_\tau(X_{\tau-}) + \1_{\tau =\tau^C \leq  T}\Jt^C_\tau(X_{\tau-}). \nonumber 
\end{align}
where $\Jt^i \coloneqq B^{-1}\Theta^i$, $i \in \{H, C\}$. 
For example, under \textit{replacement cost}:
\begin{align}
 \Jt^H_t(X_{t-})=&-X_{t-} +L^HL^m(X_{t-} + \pt_{t-})^+, \label{def:JH}\\
 \Jt^C_t(X_{t-})=&X_{t-} +L^CL^m(X_{t-} + \pt_{t-})^-. 
\end{align}
while under \textit{clean price}, $\Jt^H_t(X_{t-})=L^HL^m\pt_{t-}^+$ and
$\Jt^C_t(X_{t-})=L^CL^m\pt_{t-}^-$.  In the case of \textit{clean price},
$\Jt^i$, $i \in \{H, C\}$, represent the loss inflicted to the hedger from breaching
of the contract.  Recall that $B$ and $P$ are independent
of $V$. Therefore, providing existence and uniqueness of the hedger's
hedging portfolio and hedging strategy, $(V, \pi)$, reduces to investigating the
BSDE of the adjustment process (\ref{dyn-xva}). Before examining the
solvability, assuming the existence and integrability, we define each component
in the incremental adjustment and provide some corresponding remarks.
\begin{definition}\label{def:xva} 
  Assume that there exists $(X, \pit$) satisfying (\ref{dyn-xva}). Then for
  $t < \btau$, we define adjustment processes: $ \emph{FCA}$, $\emph{FBA}$, $\emph{CVA}$, $\emph{DVA}$, and incremental
  funding adjustment processes: $ \emph{FBA}^\Delta$, $\emph{FCA}^\Delta$  as follows:  
\begin{align}
  \emph{FCA}_t \coloneqq
  & \dsE\bigg[\int_{t}^{\btau}\big(X_s + \pt_s +\Bt^\epsilon_s- c_s
    - \SUM_{i \in I \setminus \rho} \tilde{\pi}^i_s\big)^{-}s^{b}_s \df s\bigg\vert\calG_t\bigg],\label{fca}\\
  \emph{FBA}_t \coloneqq
  & \dsE\bigg[ \int_{t}^{\btau}\big(X_s + \pt_s +\Bt^\epsilon_s- c_s
    - \SUM_{i \in I \setminus \rho} \tilde{\pi}^i_s\big)^{+}s^{\ell}_s \df s\bigg\vert\calG_t\bigg], \label{fba}\\
  \emph{DVA}_t \coloneqq
  & \dsE\bigg[\1_{\btau = \tau=\tau^H}\Jt^H_{\tau}(X_{\tau-})\bigg\vert\calG_t\bigg],\label{dva}\\
  \emph{CVA}_t \coloneqq
  & \dsE\bigg[\1_{\btau =
    \tau=\tau^C}\Jt^C_{\tau}(X_{\tau-})\bigg\vert\calG_t\bigg], \label{cva}\\
    \calO_t \coloneqq&
    \Big[\int_{t}^{\btau}s^\epsilon_s\Bt^\epsilon_s \df s\big|\calG_t\Big], \label{opp.cost}\\
  \emph{FCA}^\Delta_t \coloneqq
  & \emph{FCA}_t - \calO_t^-,\label{fcad}\\
  \emph{FBA}^\Delta_t \coloneqq
  & \emph{FBA} - \calO_t^+, \label{fbad}
\end{align}
where \cref{fca}-\cref{opp.cost} are well-defined. In this case, we also define 
\emph{FVA} and $\emph{FVA}^\Delta$ as follows:
\begin{align}
    \emph{FVA}\coloneqq& \emph{FCA}-\emph{FBA}, \nonumber \\
  \emph{FVA}^\Delta \coloneqq& \emph{FCA}^\Delta-\emph{FBA}^\Delta. \nonumber 
\end{align}
\end{definition} 
\begin{remark}
\begin{enumerate}[(i)]
\item Assume that the local-martingales in (\ref{dyn-xva}) are true
  martingales. Then:
  \begin{align}
    X= \text{FVA}^\Delta-\text{DVA}+\text{CVA} 
    =\text{FCA}^\Delta-\text{FBA}^\Delta-\text{DVA}+\text{CVA}. \label{remark.xva.decomp}     
  \end{align}
\item Under \textit{replacement cost}:
  \begin{align}
    \text{FCA}^\Delta_t =
    & \dsE\bigg[\int_{t}^{\btau}\Big[\big(L^mX_s + L^m\pt_s +\Bt^\epsilon_s
      - \SUM_{i \in I \setminus \rho} \tilde{\pi}^i_s\big)^{-}s^{b}_s - (s^\epsilon_s\Bt^\epsilon_s)^-\Big] \df s\bigg\vert\calG_t\bigg],\label{replacement-fca}\\
    \text{FBA}^\Delta_t =
    & \dsE\bigg[ \int_{t}^{\btau}\Big[\big(L^mX_s + L^m\pt_s +\Bt^\epsilon_s- \SUM_{i \in I \setminus \rho} \tilde{\pi}^i_s\big)^{+}s^{\ell}_s - (s^\epsilon_s\Bt^\epsilon_s)^+\Big]\df s\bigg\vert\calG_t\bigg].      \label{replacement-fba}
  \end{align}
  It is frequently stated that there is no FVA when contracts are fully
  collateralized, \citep[e.g., see][]{cameron2013black}. To examine this, assume
  full collateralization, i.e., $L^m=0$, \textit{replacement cost}, and
  $\rho=I$. We can then see that $\FCAd = \text{FBA}^\Delta = 0$ from
  (\ref{replacement-fca}) and (\ref{replacement-fba}). Therefore, based on our
  model, $\FVAd=0$  when the close-out amount is the
  \textit{replacement cost}. However, under \textit{clean price}, $\FVAd$
  still exists even when $L^m=0$. This is one of the reasons why
  \textit{replacement cost} should be discussed.
\item Considering different lending/borrowing not only makes the BSDEs for
  replication pricing semi-linear, but also makes the associated Hamiltonians
  non-smooth in optimal investment problems \citep[see][for
  examples]{bo2017portfolio, bo2016optimal, yang2019constrained}.   
\end{enumerate}
\end{remark}
In the subsequent section, we represent (\ref{dyn-xva}) as a standard form, and reduce
it to a BSDE on the reference filtration $\dsF$. 
\subsection{Filtration Reduction}
\label{sec:bsde}
To represent $(X, \pit)$ in \cref{dyn-xva} as a solution of a BSDE, we need to
specify its generator. To achieve this, let us  put $Z\coloneqq
\pit^\top\sigma - Z^P$, and consider a map $(\phi_t)_{t\geq0}$ such that
 for any $t \geq 0 $:
\begin{align}
\phi_t(X_t, Z_t) = \SUM_{i \in I \setminus \rho}\tilde{\pi}^i_t. \nonumber    
\end{align}
The form of $\phi$ varies depending on the considered financial markets because
$\sigma$ may not be invertable. As will be shown in \Cref{example:generator} and
\Cref{lemma:reduction}, $\phi$ is linear in $Z, Z^P, \pit^H, \pit^C$, but
$\tilde{\pi}^i$, $i \in \{H, C\}$, depend on a solution of associated BSDEs
under \textit{replacement cost}.  Therefore, for simplicity, we assume the following:
\begin{assumption}\label{assm:generator}
  There exists an $n$-dimensional $\dsF$-progressively measurable process
  $\alpha$ such that:
 \begin{align}
   \phi_t(z) = \alpha^\top_t(z +  Z^P_t).\nonumber 
 \end{align}
\end{assumption}
Dependence on $\pit^i$, $i\in \{H, C\}$ makes the calculation
more complicated. In \Cref{app:sto.int}, we report one example in which 
\Cref{assm:generator} is not satisfied. Now, by the following examples, we will
demonstrate that \Cref{assm:generator} does not impose a  strong restriction on our
considered financial markets.
\begin{example}\label{example:generator}
  (\rom{1}) If $\rho = I$, we should trivially set 
  $\phi_t (z)= 0$.
  
(\rom{2}) Consider $n=1$, and constant parameters. Then $S^1,~S^H,~S^C $ follow:
\begin{align}
  \df S^1_t =& r S^1_t \df t+\sigma^1 S^1_t\df W_t, \nonumber\\
  \df S^i_t =& r S^i_t \df t - S_{t-}\df M\IT,~~i \in \{ H, C\}. \nonumber
\end{align}
It then follows that
 $\sum_{i \in I}\tilde{\pi}^i\sigma^i = \pit^1 \sigma^1$.
When $\rho = \{H, C\}$:
\begin{align}
  \phi_t(z)=  (\sigma^1)^{-1}(z + Z^P_t) = (\sigma^1)^{-1}z +
  (\sigma^1)^{-1}Z^P_t.  \label{trans1} 
\end{align}
On the other hand, when $\rho =\{1, C\}$, we have $\sum_{i \in I\setminus\rho}\pit^i = \pit^H$. Therefore,
  $\phi_t \colon z \to   \pit^H_t$. 
This is the same case discussed by \cite{burgard2010partial}. 

(\rom{3}) Let us assume that the OIS rate is an $\dsF$-adapted
process. In addition, we assume that for any $i \in \{H, C\}$,
$(G^i_t)_{t \geq0}$ is given by
    $\df G\IT = -h\IT G\IT \df t$,    
  where $(h\IT)_{t\geq0}$ are deterministic processes.  We consider non-defaultable
  and defaultable zero coupon bonds with the same maturity  $T$, i.e.,
  $S^1, ~S^H, ~S^C$ are defined as $ S^1_t \coloneqq  B_t\dsE[
  B_T^{-1}\mid\calF_t]$ and $S^i_t \coloneqq  B_t\dsE[ \1_{\tau^i > T }B_T^{-1}\mid\calG_t], ~~i \in \{H, C\}$. 
By \Cref{lem:red} in \Cref{app:red},
$ S^i_t = \1_{t < \tau^i}B_t(G^i_t)^{-1}\dsE[ G^i_TB_T^{-1}|\calF_t] $.  Since
$(G^i_t)_{t \geq0}$, $i \in \{H, C\}$, are deterministic,
  $S^i_t = \1_{t < \tau^i}(G^i_t)^{-1} G^i_TS^1_t$. 
It follows that for $t < \tau$, $\sigma^1 = \sigma^H = \sigma^C$. Recall $\Sigma= [(\sigma^1)^\top \cdots
(\sigma^n)^\top]^\top$ and $\sigma = [\Sigma^\top~ (\sigma^H)^\top~ (\sigma^C)^\top]^{\top}$. Then,
  $(\Sigma^\top)^{-1}\sigma^\top\pit = 
  [\pit^1+\pit^H+\pit^C, ~ \pit^2 , ~\cdots \pit^n ]^\top$.
Therefore, $\1^\top(\Sigma^\top)^{-1}\sigma^\top\pit = \sum_{i \in I}\pit$, where
$\1 \coloneqq (1, \dots,1)^\top \in\dsR^n$. Thus, if $\rho = \emptyset$:
\begin{align}
 \phi_t\colon z \to  \1^\top(\Sigma^\top_t)^{-1}(z + Z^P_t), \label{trans3}
\end{align}
Now, consider $\rho \not=\emptyset $ and  define $\1^\rho \in \dsR^n$ by:
\begin{align}
  (\1^\rho)_i\coloneqq
\left\lbrace
  \begin{array}{r@{}l}
    0 ~~ & i \in \rho,\\
    1~~  & i\notin \rho,
  \end{array}
  \right. \nonumber  
\end{align}
where $(\1^\rho)_i$ denotes the $i$-th component of $\1^\rho$. When
$\rho\cap \{1, H, C\} = \emptyset$,
$\phi_t\colon z \to (\1^\rho)^\top(\Sigma^\top_t)^{-1}(z + Z^P_t)$. However, if $\rho = \{1\}$,
  $\phi_t\colon z \to  (\1^\rho)^\top(\Sigma^\top_t)^{-1}(z + Z^P_t) +\pit^H + \pit^C$.  \qed
\end{example}
Now, we denote the generator of (\ref{dyn-xva}) by $g^\dsG$:
\begin{align}
  g^{\dsG}_t(y, z) \coloneqq
  -\big(y + \pt_t +\Bt^\epsilon_t- c_t - \phi_t(z)\big)^{+}s^{\ell}_t
  +\big(y + \pt_t +\Bt^\epsilon_t- c_t - \phi_t(z)\big)^{-}s^{b}_t +s^\epsilon_t\Bt^\epsilon_t. \nonumber 
\end{align}
Let $(Y^{\dsG}, Z^{\dsG}, \pit^H, \pit^C)$ denote the solution, in  certain
spaces, of the following BSDE:
\begin{align}
  Y^{\dsG}_t =&-\1_{\btau=\tau=\tau^H}\Jt^{H}_{\tau}(Y^{\dsG}_{\tau-})
                +\1_{\btau=\tau=\tau^C}\Jt^{C}_{\tau}(Y^{\dsG}_{\tau-})
                \nonumber\\
              &  +\int_{t}^{\btau}g^{\dsG}_s( Y^{\dsG}_s, Z^{\dsG}_s) \df s
                - \int_{t}^{\btau}(Z^{\dsG}_s)^{\top}\df W_s
                +\SUM_{i =H, C}\int_{t}^{\btau} \tilde{\pi}^i_s\df M^i_s. \label{GBSDE}
\end{align}
Then $(Y^{\dsG}, Z^{\dsG}, \pit^H, \pit^C)$ provides $(X, \pit)$ as well as
$(V, \pi)$. Instead of dealing directly with (\ref{GBSDE}), however, we will
investigate a reduced BSDE on the reference filtration $\dsF$. The idea is as
follows:

It is an established fact that, in the progressively enlarged filtration $\dsG$, any
$\dsG$-optional (resp. predictable) process has an $\dsF$-optional
(resp. predictable) reduction. Therefore, if there exists a solution of
(\ref{GBSDE}) such that $Y^\dsG$ is $\dsG$-optional and $Z^\dsG$ is
$\dsG$-predictable, there exists an $\dsF$-adapted pair $(Y^\dsF, Z^\dsF)$
that satisfies the following:
\begin{align}
 Y^\dsG =& \1_{t< \btau}Y^\dsF, ~~~~~Z^\dsG = \1_{t \leq \btau}Z^\dsF.\label{yred}  
\end{align}
Moreover, we surmise that $(Y^\dsF, Z^\dsF)$ is a solution of a BSDE on the
reference filtration $\dsF$, i.e., we will determine $g^\dsF \colon \Omega \times [0, T] \times
\dsR^{n+1}\to \dsR$ such that:
\begin{align}
 Y^{\dsF}_t = \int_{t}^{T}g^{\dsF}_s(Y^{\dsF}_s, Z^{\dsF}_s) \df s - \int_{t}^{T}(Z^{\dsF}_s)^{\top}
  \df W_s. \label{FBSDE}
\end{align}
Then, by the terminal condition
$Y^\dsG_{\btau} =
-\1_{\btau=\tau=\tau^H}\Jt^{  H}_{\tau}(Y^{\dsG}_{\tau-})+\1_{\btau=\tau=\tau^C}\Jt^{ C}_{\tau}(Y^{\dsG}_{\tau-})$,
together with (\ref{yred}), 
we can expect that:
\begin{align}
  Y^{\dsG}_t = \1_{t <\btau}Y^{\dsF}_t + \1_{t \geq \btau}
  \big[-\1_{\btau=\tau=\tau^H}\Jt^{ H}_{\tau}(Y^{\dsF}_{\tau-})+\1_{\btau=\tau=\tau^C}\Jt^{ C}_{\tau}(Y^{\dsF}_{\tau-})
  \big]. \label{eq:fg}
\end{align}
By applying It\^o's formula to (\ref{eq:fg}), we can determine 
$g^\dsF$. The next lemma explains the corresponding  details. We report
the proof in \Cref{app:lemmas}. By the next lemma, $(Y^\dsG, Z^\dsG)$ can be
obtained by $(Y^\dsF, Z^\dsF)$. Then, $(X, \pi)$ and $(V, \pi)$ can be found by $(Y^\dsG, Z^\dsG)$.
\begin{lemma}\label{lemma:reduction}
Assume that there exists a unique solution $(Y^\dsF, Z^\dsF) \in \dsS^{2}_T \times \dsH^{2, n}_T$ of the following BSDE:
\begin{align}
  &Y^{\dsF}_t = \int_{t}^{T}g^{\dsF}_s(Y^{\dsF}_s, Z^{\dsF}_s) \df s - \int_{t}^{T}(Z^{\dsF}_s)^{\top}
                \df W_s,    \label{FBSDE2} \\
 &g^\dsF_t(y, z) \coloneqq g^\dsG_t(y, z) - h^H_t\Jt^{H}_t(y) +
                           h^C_t\Jt^{C}_t(y)-h_ty.
                           \label{fgenerator}
\end{align}
Then $(Y^\dsG, Z^\dsG)$ can solve (\ref{GBSDE}) by the following relationships:
\begin{align}
  Y^{\dsG}_t =
  & \1_{t <\btau}Y^{\dsF}_t + \1_{t \geq \btau} \big[-\1_{\btau=\tau=\tau^H}
    \Jt^{H}_{\tau}(Y^{\dsF}_{\tau-})+\1_{\btau=\tau=\tau^C}\Jt^{C}_{\tau}(Y^{\dsF}_{\tau-}) \big], \label{fsol1}\\
  Z^\dsG_t =& \1_{t < \btau}Z^\dsF_t,\\
  \pit^H_t =& Y^\dsF_t + \Jt^H_t(Y^\dsF_{t-}), ~~~\pit^C_t = Y^\dsF_t -
              \Jt^C_t(Y^\dsF_{t-})\label{fsol3} 
 \end{align}
Furthermore, if (H)-hypothesis holds between $\dsF$ and $\dsG$, i.e., any  $(\dsF,
\dsQ)$-martingale is a $(\dsG, \dsQ)$-martingale, then (\ref{fsol1})-(\ref{fsol3}) is
the only solution of (\ref{GBSDE}).
\begin{remark}
This reduction argument was also used by \cite{bichuch2017arbitrage,
  crepey2015bilateral1, crepey2015bilateral2, brigo2016analysis}.   
\end{remark}
\end{lemma}
Nevertheless, due to the semi-linearity in \cref{FBSDE2} and \cref{fgenerator},
we need to solve (\ref{FBSDE}) numerically, and thus one may be interested in
how much offset between FCA and FBA can be anticipated. The answer is negative
according to our main results given in the next section. It will be shown that
the funding impact is binary, i.e., either FCA$=0$ or FBA$=0$, and the binary
impact is determined by the monotone property of the payoff function.  The
result is also related to the conundrum regarding FBA/DVA, i.e., whether FBA
double-counts DVA \cite[see][]{cameron2013black}.  However, our results inform us
that FBA and DVA are affected by different mathematical structures of the
derivative contract. More precisely, DVA occurs where the bank has liabilities,
i.e., DVA concerns the sign of the payoff function, while FBA concerns
the monotone property. A detailed explanation of this is provided in
\Cref{main.remark}.

\section{Main Results} 
\label{sec:main}
In this section, we demonstrate the binary nature of FVA: FVA is either FCA or -FBA. In
other words, either $\text{FCA}=0$ or $\text{FBA}=0$. This switching property of
FVA is determined by certain properties of payoff structures of the derivative
contract. Before the main theorems, readers may want to refer to \Cref{sec:idea}
where the idea for the proof is explained with an example.
 
We consider deterministic default intensities, volatility, and funding spreads,
but $(r_t)_{t\geq0}$ does not need to be deterministic so that we can apply the
result to cases in which $(r_t)_{t\geq0}$ is a general $\dsF$-adapted
process. These assumptions are made primarily for simplicity. One case of
stochastic default intensities is reported in \Cref{app:sto.int}.  We also 
consider derivatives of European style.  For derivatives that have cash-flows at
multiple times, we can divide the interval $[0, T]$ according to the time of
cash-flows. For example, if $\frD_t = \sum_{i=1}^N\1_{T_i \leq t}\xi^i$, for
$t \in (T_{i-1}, T_i)$, we consider the following BSDE:
\begin{align}
  Y^\dsF_t = Y^\dsF_{T_i} + \int_{t}^{T_i}g^\dsF_s(Y^\dsF_s , Z^\dsF_s) \df s
  +\int_{t}^{T_i}(Z^\dsF_s)^\top \df W_s. \label{divide}
\end{align}
Then we can apply the next theorems for each BSDE (\ref{divide}). 
 We start from \textit{clean
  price}. 
\begin{theorem}[\textit{Clean price}]\label{thm:clean}
We assume  that $h^H$, $h^C$, $s^\ell$, $s^b$ are both deterministic and
bounded. Moreover, $\alpha$ is also deterministic and $\alpha s^\ell$, $\alpha s^b$ are
bounded. We consider \textit{clean price} as $e = P$, and
$B^{-1}\frD = \1_{\llbracket T, \infty\llbracket}\xi$. In addition, we assume the following:
\begin{eqnarray}
  & \xi \in\dsL_T^2 \cap\dsD^{1, 2},~~ D_\theta\xi \in (\dsD^{1,2})^n~~\text{for any } \theta \leq T,
    \label{thm1:int1}\\
  &\dsE\bigg[\int_{0}^{T}\big|D_\theta \xi\big|^2\df \theta\bigg] <\infty, \label{thm1:int2}\\
  &\dsE\bigg[\int_{0}^{T}\int_{0}^{T}\big|D_t(D_\theta \xi)\big|^2\df \theta \df t\bigg] <\infty.\label{thm1:int3}
\end{eqnarray}
We assume that $\1_{\pt=0} = 0$, $\df \dsQ \otimes \df t$ a.s. Let:
\begin{align}
\xi^b \coloneqq& s^b(1-L^m) \alpha^\top Z^P  +h\Bt^\epsilon+(h-h^HL^HL^m)\pt^+
  -(h-h^CL^CL^m)\pt^-, \nonumber \\
\xi^\ell \coloneqq& s^\ell(1-L^m) \alpha^\top Z^P  +h\Bt^\epsilon+(h-h^HL^HL^m)\pt^+
  -(h-h^CL^CL^m)\pt^-.\nonumber 
\end{align}

(\rom{1}) Assume that for any $\theta \leq T$:
\begin{align}
    L^m(\xi - \alpha_\theta^\top D_\theta\xi) + \Bt^\epsilon_T \leq&0, ~~\text{a.s},\label{thm1:cx}\\
  \alpha^\top Z^P \geq&0, ~~\dqdt,   \label{thm1:czf}\\
  \xi^b - \alpha_\theta^\top D_\theta\xi^b \leq&0, ~~\dqdt,\label{thm1:cxb}
\end{align}
then there exists $(Y^\#, Z^\#) \in \dsS^2_T \times\dsH^{2, n}_T$ that satisfies:
\begin{align}
  Y^\#_t =& \int_{t}^{T} -\bigg[\Big(Y^\#_s +L^m\pt_s+\Bt^\epsilon_s-\phi_s(Z^\#_s)\Big)s^{b}_s-s^\epsilon_s\Bt^\epsilon_s+h_sY^\#_s\bigg]\df s\nonumber\\
  &+\int_{t}^{T}
    -\big[\1_{\pt_s\geq0}h^H_sL^HL^m+\1_{\pt_s<0}h^C_sL^CL^m\big]\pt_s\df s
    -\int_{t}^{T}(Z^\#_s)^\top\df W_s, \nonumber 
\end{align}
and $(Y^\#, Z^\#) = (Y^\dsF, Z^\dsF)$. In particular, \emph{FBA}=$0$.

(\rom{2}) Assume that for any $\theta \leq T$:
\begin{align}
    L^m(\xi - \alpha_\theta^\top D_\theta\xi) +\Bt^\epsilon_T \geq&0,~~\text{a.s,}\label{thm:cx}\\
  \alpha^\top Z^P \leq&0, ~~\dqdt,   \label{thm:czf2}\\
  \xi^\ell - \alpha_\theta^\top D_\theta\xi^\ell \geq&0,~~\dqdt, \label{thm:cxl}
\end{align}
then there exists $(Y^\#, Z^\#) \in \dsS^2_T \times\dsH^{2, n}_T$ that satisfies:
\begin{align}
  Y^\#_t =& \int_{t}^{T} -\bigg[\Big(Y^\#_s +L^m\pt_s+\Bt^\epsilon_s-\phi_s(Z^\#_s)\Big)s^{\ell}_s-s^\epsilon_s\Bt^\epsilon_s+h_sY^\#_s\bigg]\df s\nonumber\\
  &+\int_{t}^{T}
    -\big[\1_{\pt_s\geq0}h^H_sL^HL^m+\1_{\pt_s<0}h^C_sL^CL^m\big]\pt_s\df s
    -\int_{t}^{T}(Z^\#_s)^\top\df W_s, \nonumber 
\end{align}
and $(Y^\#, Z^\#) = (Y^\dsF, Z^\dsF)$. In particular, $\emph{FCA}=0$.

(\rom{3}) If the contract is un-collateralized, i.e., $L^m = 1$, the results 
in (\rom{1}) and  (\rom{2}) hold without (\ref{thm1:czf}) and (\ref{thm:czf2}).
\end{theorem}
\begin{remark}\label{main.remark}
  \begin{enumerate}[(i)]
  \item The conditions \cref{thm1:cx}-\cref{thm:cxl} are closely related to the
    monotone property of $\xi$, and this shows us how FBA and DVA are dissimilar.
    DVA is a benefit from the possibility that the bank may default on its
    derivative payables. Therefore, DVA occurs when $\xi \geq0$. On the other hand, FBA
    occurs when \cref{thm:cx}-\cref{thm:cxl} are satisfied. To elucidate the
    meaning of the inequalities, we consider the same financial market as in
    \Cref{example:idea}. In addition, let $\xi = \psi(\St^1_T)$ for some smooth
    function $\psi \colon \dsR \to \dsR_+$. Then, \cref{thm:czf2} becomes
    $ \psi'(\St^1_T) \St^1_T \leq0$, which means that $\psi$ is non-increasing. This also
    implies that:
  \begin{align}
    \xi -\alpha_\theta^\top D_\theta \xi = \psi(\St^1_T) - \psi'(\St^1_T) \St^1_T \geq0. \label{main.remark:cond1}
  \end{align}
  The last inequality \cref{thm:cxl} is also satisfied under a similar monotone
  property.
  This will be later shown with examples. 
\item One reason for believing that FBA double-counts DVA may be the
  convention in the literature that all assets can be traded in repo markets, i.e.,
  $\rho = I$. Recall that $\alpha = 0$ when $\rho = I$ (recall (\rom{1}) in
  \Cref{example:generator}). When $\alpha = 0$, by \cref{main.remark:cond1}, FBA also
  occurs when $\xi \geq0$. However, the existence of repo markets is a strong
  condition since we can not guarantee that repo markets are always available.
  \end{enumerate}
\end{remark}
\begin{proof}[Proof of Theorem \ref{thm:clean}]We only prove
  (\rom{1}). (\rom{2}) can be proved similarly, and (\rom{3}) is an obvious
  consequence of (\rom{1}) and (\rom{2}).
  
(\rom{1}) By (\ref{thm1:int1}) and (\ref{thm1:int2}):
\begin{align}
  \dsE\bigg[ \int_{0}^{T}\big|\pt _t\big|^2\df t\bigg]
  & \leq
     \int_{0}^{T}\dsE[|\xi|^2]\df t<\infty,\nonumber \\
  \dsE\bigg[ \int_{0}^{T}\big|Z _t\big|^2\df t\bigg]
  & \leq
     \int_{0}^{T}\dsE[|D_t\xi|^2]\df t<\infty.\nonumber 
\end{align}
One can easily  see that there exists a unique solution $(Y^\#,Z^\#) \in \dsS^2_T
\times \dsH^{2, n}_T$ of the following BSDE:
\begin{align}
  Y^\#_t =& \int_{t}^{T}g^\#_s(Y^\#_s, Z^\#_s)\df s    -\int_{t}^{T}(Z^\#_s)^\top\df W_s,
            \nonumber \\
  g^\#_t(y, z):=&-\big(y +L^m\pt_t+\Bt^\epsilon_s-\phi_t(z)\big)s^{b}_t
  +s^\epsilon_s\Bt^\epsilon_s
  -h_ty  -h^H_tL^HL^m(\pt_t)^++h^C_tL^CL^m(\pt_t)^-. \nonumber 
\end{align}
We will show that:
\begin{align}
(Y^\#, Z^\#) = (Y^\dsF, Z^\dsF), \label{thm1:purpose}
\end{align}
where $(Y^\dsF, Z^\dsF)$ is the solution of (\ref{FBSDE2}). Because $(Y^\dsF,
Z^\dsF)$ satisfying (\ref{FBSDE}) is unique in $\dsS^2_T \times \dsH^{2, n}_T$, to prove
(\ref{thm1:purpose}), it suffices to show that:
\begin{align}
 Y^\# + L^m\pt +\Bt^\epsilon- \phi(Z^\#) \leq0 , ~~~\dqdt.  \label{thm1:ineq-s}
\end{align}
To this end, we
introduce another transformation:
\begin{align}
  V^\dsF \coloneqq Y^\# + L^m\pt+\Bt^\epsilon,
  ~~~~\Pi^\dsF \coloneqq Z^\# + L^mZ^P. \nonumber 
\end{align}
Then, by (\rom{3}) in \Cref{lem:pf}, $(V^\dsF, \Pi^\dsF)$ satisfies the following:
\begin{align}
  V^\dsF = L^m\xi +\Bt^\epsilon_T +\int_{t}^{T}F_s(V^\dsF_s, \Pi^\dsF_s)\df s
  - \int_{t}^{T}(\Pi^\dsF_s)^\top\df W_s,\nonumber 
\end{align}
where:
\begin{align}
  F_t(y, z) \coloneqq
  &  g^{\#}_t(y - L^m\pt_t-\Bt^\epsilon, z -  L^mZ^P_t)\nonumber\\
  =&-(y - \alpha^\top_tz)s^b_t - h_ty\nonumber\\
  &+ (1-L^m)s^b_t\alpha^\top_t Z^P_t 
    +h_t\Bt^\epsilon_t+(h_t-h^H_tL^HL^m)(\pt_t)^+-(h_t-h^C_tL^CL^m)(\pt_t)^- \nonumber\\
  =&-(y - \alpha^\top_tz)s^b_t- h_ty + \xi^b_t.\nonumber 
\end{align}
Note that (\ref{thm1:ineq-s}) is equivalent to:
\begin{align}
  V^\dsF - (1-L^m)s^b\alpha^\top Z^P  - \phi(\Pi^\dsF)\leq 0, ~~~\dqdt. 
  \label{thm1:ineq-ss}
\end{align}
To show (\ref{thm1:ineq-ss}), we use Malliavin calculus and the comparison principle
of BSDEs.

By (\ref{thm1:int2}) and (\ref{thm1:int3}):
\begin{align}
  \dsE\bigg[ \int_{0}^{T}\int_{0}^{T}\big|D_\theta\pt_t\big|^2\df t \df \theta\bigg]
  &\leq
     \int_{0}^{T}\int_{0}^{T}\dsE|D_\theta\xi|^2\df t \df \theta<\infty,\nonumber \\
  \dsE\bigg[ \int_{0}^{T}\int_{0}^{T}\big|D_\theta Z^P_t\big|^2\df t \df \theta\bigg]
  &\leq
     \int_{0}^{T}\int_{0}^{T}\dsE|D_\theta(D_t\xi)|^2\df t \df \theta<\infty.\nonumber 
\end{align}
Therefore, by Proposition 5.3 in \cite{el1997backward}, $(V^\dsF, \Pi^\dsF) \in
\dsL^2([0, T]:  \dsD^{1,2}\times (\dsD^{1,2})^n)$, and for any $1\leq i \leq n$, a
version of $\{(D^i_\theta V^\dsF_t, D^i_\theta \Pi^\dsF_t)| ~0\leq\theta, t \leq T\}$ is given by the following:
\begin{align}
  D^i_\theta V^\dsF_t=
  & L^mD^i_\theta \xi + \int_{t}^{T}\Big[ -(D^i_\theta V^\dsF_s
    - \alpha^\top_s D^i_\theta \Pi^\dsF_s)s^b_s -h_s D^i_\theta V^\dsF_s
    +D^i_\theta \xi^b_s \Big] \df s -\int_{t}^{T}(D^i_\theta \Pi^\dsF_s)^\top \df W_s, 
\end{align}
and $\{D_tV^\dsF_t\colon~0\leq t\leq T\}$ is a version of $\{\Pi^\dsF_t\colon~0\leq t\leq T\}$.
Let us denote:
\begin{align}
  V^\dsF_{t, \theta}\coloneqq& \alpha ^\top _\theta (D_\theta V^\dsF_t), ~~~~
  \Pi^\dsF_{t, \theta}\coloneqq (D_\theta \Pi^\dsF_t)\alpha_\theta.\nonumber 
\end{align}
Then $(V^\dsF_{t, \theta}, \Pi^\dsF_{t, \theta})$ is given by:
\begin{align}
V^\dsF_{t, \theta} = L^m\alpha _\theta D_\theta \xi+\int_{t}^{T}\Big[F_s(V^\dsF_{s,
  \theta},\Pi^\dsF_{s,\theta})
  +\alpha^\top_\theta D_\theta \xi^b_s - \xi^b_s  \Big]\df s - \int_{t}^{T}(\Pi^\dsF_{s,\theta})^\top\df
  W_s.\nonumber 
\end{align}
Therefore, by (\ref{thm1:cx}) and (\ref{thm1:cxb}) together with the comparison
principle of BSDEs, we determine that $V^\dsF_{t, \theta} \geq V^\dsF_t$ for any $\theta \leq
t$. Moreover, by (\ref{thm1:czf}), for any $\theta \leq t$:
\begin{align}
 V^\dsF_t  -(1-L^m)s^b_tZ^P_t-V^\dsF_{t, \theta}\leq V^\dsF_t  -V^\dsF_{t, \theta}
  \leq0,\nonumber 
\end{align}
which implies (\ref{thm1:ineq-s}) and (\ref{thm1:ineq-ss}). Furthermore, by
\cref{fba}, FBA$=0$. 
\end{proof}
We now turn to the cases of \textit{replacement cost}. Recall that when we
consider \textit{replacement cost}, $\Jt^{ i}$, $i \in \{H, C\}$, depend on
$Y^\dsF$. However, $\Jt^{i}(y)$ is not differentiable in $y$. We can avoid
this irregularity by considering contracts such that either $\pt \geq0$ or
$\pt \leq0$, $\df \dsQ \otimes \df t$ a.s, i.e., options.
\begin{theorem}[\textit{Replacement cost}] \label{thm:replacement} We
  assume that $h^H$, $h^C$, $s^\ell$, $s^b$ are both deterministic and bounded.  Moreover,
  $\alpha$ is also deterministic and $\alpha s^\ell$, $\alpha s^b$ are bounded. We consider
  \textit{replacement cost} as $e = V_- -B^\epsilon = B Y^\dsF + P$, and
  $B^{-1}\frD = \1_{\llbracket T, \infty\llbracket}\xi$. In addition, we assume the following:
\begin{eqnarray}
  & \xi \in\dsL_T^2\cap \dsD^{1, 2}, \label{thm2:int1}\\
  &\dsE\bigg[\int_{0}^{T}\big|D_\theta \xi\big|^2\df \theta\bigg] <\infty, \label{thm2:int2}
\end{eqnarray}
and that either $\xi \geq0$ or $\xi \leq0$ a.s, i.e., we consider options.

(\rom{1}) Let  $\epsilon \leq0$ and for any $\theta \leq T$:
\begin{align}
  L^m\xi - \alpha_\theta^\top D_\theta\xi \leq0, ~~\text{a.s.} \label{thm2:condition}
\end{align}
If $ \xi \geq0 $ a.s, then there exists a solution
$(Y^\#, Z^\#)\in \dsS^2_T \times \dsH^{2, n}_T$ that satisfies:
\begin{align}
  Y^\#_t =& \int_{t}^{T} -\bigg[\Big(L^mY^\#_s +L^m\pt_s+\Bt^\epsilon_s-\phi_s(Z^\#_s)\Big)s^{b}_s-s^\epsilon_s\Bt^\epsilon_s\bigg]\df s\nonumber\\
  &+\int_{t}^{T}
    \big[-h^H_sL^HL^m(Y^\#_s+\pt_s)\big]\df s
    -\int_{t}^{T}(Z^\#_s)^\top\df W_s, \nonumber 
\end{align}
and $(Y^\#, Z^\#)=(Y^\dsF, Z^\dsF)$. On the other hand, if $\xi \leq0 $ a.s, then
there exists a solution $(Y^\#, Z^\#)\in \dsS^2_T \times \dsH^{2, n}_T$ that satisfies:
\begin{align}
  Y^\#_t =& \int_{t}^{T} -\bigg[\Big(L^mY^\#_s +L^m\pt_s+\Bt^\epsilon_s-\phi_s(Z^\#_s)\Big)s^{b}_s-s^\epsilon_s\Bt^\epsilon_s\bigg]\df s\nonumber\\
          &+\int_{t}^{T}
            \big[-h^C_sL^CL^m(Y^\#_s+\pt_s)
            \big]\df s
            -\int_{t}^{T}(Z^\#_s)^\top\df W_s, \nonumber 
\end{align}
and $(Y^\#, Z^\#)=(Y^\dsF, Z^\dsF)$. In particular, for both cases, $\emph{FBA}=0$. 

(\rom{2}) Assume that $\epsilon \geq0$ and for any $\theta \leq T$:
\begin{align}
  L^m\xi - \alpha_\theta^\top D_\theta\xi \geq0, ~~\text{a.s.} \label{thm2:cond}
\end{align}
If $\xi \geq0 $ a.s, then there exists a solution
$(Y^\#, Z^\#)\in \dsS^2_T \times \dsH^{2, n}_T$ that satisfies:
\begin{align}
  Y^\#_t =& \int_{t}^{T} -\bigg[\Big(L^mY^\#_s +L^m\pt_s+\Bt^\epsilon_s-\phi_s(Z^\#_s)\Big)s^{\ell}_s-s^\epsilon_s\Bt^\epsilon_s\bigg]\df s\nonumber\\
  &+\int_{t}^{T}
    \big[-h^H_sL^HL^m(Y^\#_s+\pt_s)\big]\df s
    -\int_{t}^{T}(Z^\#_s)^\top\df W_s, \nonumber 
\end{align}
and $(Y^\#, Z^\#)=(Y^\dsF, Z^\dsF)$. On the other hand, if $\xi \leq0 $ a.s, then there exists a solution $(Y^\#, Z^\#)\in \dsS^2_T \times \dsH^{2, n}_T$ that satisfies:
\begin{align}
  Y^\#_t =& \int_{t}^{T} -\bigg[\Big(L^mY^\#_s +L^m\pt_s+\Bt^\epsilon_s-\phi_s(Z^\#_s)\Big)s^{\ell}_s-s^\epsilon_s\Bt^\epsilon_s\bigg]\df s\nonumber\\
  &+\int_{t}^{T}
    \big[-h^C_sL^CL^m(Y^\#_s+\pt_s)\big]\df s
    -\int_{t}^{T}(Z^\#_s)^\top\df W_s, \nonumber 
\end{align}
and $(Y^\#, Z^\#)=(Y^\dsF, Z^\dsF)$. In particular, for both cases, $\emph{FCA}=0$. 
\begin{proof}
  As in the proof of Theorem \ref{thm:clean}, we can check the existence,
  uniqueness, and Malliavin differentiability of BSDEs. We only explain the
  transformation and how to apply the comparison principle. Without a loss of
  generality, we assume $\xi \geq0$. It follows that $\pt \geq0$, $\dqdt$.

(\rom{1}) We consider  a solution $(Y^\#, Z^\#) \in \dsS^2_T \times \dsH^{2, n}_T$  of the following BSDE:
\begin{align}
  Y^\#_t
  =& \int_{t}^{T} -\bigg[\big(L^mY^\#_s+L^m\pt_s +\Bt^\epsilon_s-
     \alpha^\top_s(Z^\#_s+Z^P_s)\big)s^{b}_s - s^\epsilon_s\Bt^\epsilon_s
     +h^H_sL^HL^m(Y^\#_s+\pt_s)\bigg]\df s\nonumber\\
   & -\int_{t}^{T}(Z^\#_s)^\top\df W_s.    \nonumber 
\end{align}
Take $V^\dsF \coloneqq Y^\# + \pt $, $\Pi^\dsF\coloneqq Z^\# + Z^P$. Then:
\begin{align}
  V^\dsF_t
  =& \xi+\int_{t}^{T} F_s(V^\dsF_s, \Pi^\dsF_s)\df s
    -\int_{t}^{T}(\Pi^\dsF_s)^\top\df W_s, \label{thm2:bsdes}
\end{align}
where:
\begin{align}
  \label{eq:3}
  F_t(y, z)\coloneqq -(L^my - \alpha_s^\top z)s^b  -h^H_tL^HL^my
\end{align}
Since $(0,0)$ is the unique solution of the following BSDE:
\begin{align}
 y_t
  =& \int_{t}^{T} -\bigg[\big(L^m_sy_s - \alpha_s^\top z_s\big)s^{b}_s
     +h^H_sL^HL^my_s\bigg]\df s
    -\int_{t}^{T}z^\top_s\df W_s, \label{thm2:bsdesmall}
\end{align}
by comparison between (\ref{thm2:bsdes}) and (\ref{thm2:bsdesmall}), we have
that $V^\dsF \geq0$, i.e.:
\begin{align}
 Y^\# + \pt \geq0. \label{positive}  
\end{align}
Moreover, $(V^\dsF, \Pi^\dsF) \in
\dsL^2([0, T]\colon  \dsD^{1,2}\times (\dsD^{1,2})^n)$, and for any $1\leq i \leq n$, a
version of $\{(D^i_\theta V^\dsF_t, D^i_\theta \Pi^\dsF_t)| ~0\leq\theta, t \leq T\}$ is given by the following:
\begin{align}
  D^i_\theta V^\dsF_t=
   D^i_\theta \xi + \int_{t}^{T}\Big[ -(L^mD^i_\theta V^\dsF_s
    - \alpha^\top_s D^i_\theta \Pi^\dsF_s)s^b_s -h^H_sL^HL^m D_{\theta}^iV^\dsF_s
    \Big] \df s -\int_{t}^{T}(D^i_\theta \Pi^\dsF_s)^\top \df W_s, \nonumber 
\end{align}
and $\{D_tV^\dsF_t\colon~0\leq t\leq T\}$ is a version of $\{\Pi^\dsF_t\colon~0\leq t\leq
T\}$.
Let us denote:
\begin{eqnarray}
  &V^m_t\coloneqq L^m V^\dsF_t,
                  ~~~~\Pi^m_t\coloneqq L^m \Pi^\dsF_t, \nonumber \\  
  &V^\dsF_{t, \theta}\coloneqq \alpha ^\top _\theta (D_\theta V^\dsF_t),
                          ~~~~ \Pi^\dsF_{t, \theta}\coloneqq (D_\theta \Pi^\dsF_t)\alpha_\theta.\nonumber 
\end{eqnarray}
Then $(V^m, \Pi^m)$ and $(V^\dsF_{t, \theta}, \Pi^\dsF_{t, \theta})$ are given by the following:
\begin{align}
  V^\dsF_{t, \theta} =& \alpha _\theta D_\theta \xi+\int_{t}^{T}F_s(V^\dsF_{s,
                   \theta},\Pi^\dsF_{s,\theta}) \df s - \int_{t}^{T}(\Pi^\dsF_{s,\theta})^\top\df W_s,
                   \nonumber \\
  V^m_{t} =& L^m \xi+\int_{t}^{T}F_s(V^m_s, \Pi^m_s) \df s - \int_{t}^{T}(\Pi^m_s)^\top\df
             W_s. \nonumber 
\end{align}
Therefore, by (\ref{thm2:condition}), $V^m_t \leq V^\dsF_{t, \theta}$, for any $\theta \leq
t$. It then follows that:
\begin{align}
  L^mY^\#_t+L^m\pt_t +\Bt^\epsilon_t- \alpha^\top_t(Z^\#_t+Z^P_t) 
  \leq  L^mV^\dsF_t - \alpha^\top_t\Pi^\dsF_t =L^mV^\dsF_t - V^\dsF_{t, t} \leq0.\nonumber 
\end{align}
Therefore, by uniqueness of $(Y^\dsF, Z^\dsF)$, we obtain $(Y^\dsF, Z^\dsF)=(Y^\#,
Z^\#)$. Furthermore, by \cref{fba}, $\text{FBA}=0$. The proof of (\rom{2}) is 
similar to (\rom{1}). 
\end{proof}
\end{theorem}

\section{Examples and a Closed-form Solution}
\label{sec:example}
Many standard derivatives satisfy the conditions in \Cref{thm:clean} and
\Cref{thm:replacement}. We will apply the main theorems to  several derivatives
and provide a closed-form solution for a call option. In the following, for
$i \in I$, we denote
 $\St^i \coloneqq B^{-1}S^i$. 
In addition, recall that, in the main theorems, we defined $\xi$ by
$  \xi\coloneqq B_T^{-1}\Delta\frD_T$. 

\subsection{Clean Price}
Banks purchase Treasury bonds that return less than their funding
rate. \cite{hull2012fva} firmly asserted that this demonstrates that FVA should not be
considered in derivative prices.  We will show that, when purchasing bonds, FCA$=0$
for the hedger. The financial reason for the null FCA is as follows. When the
hedger has redundant cash to purchase a bond, she may invest in either the bond or
her lending account. In other words, the opportunity cost for purchasing the bond is a  benefit coming from her lending rate. Therefore, the lending rate would
constitute
the major factor to determine the bond price, and if we assume $R^\ell = r$, as in
\cite{burgard2010partial}, the fair price for the hedger is approximately the
same as the bond price derived from discounting with the Treasury rate.  The
argument above will be verified rigorously in the next example using our main
theorem.  Recall that, in the main theorem, we only assume that $s^\ell$ and
$s^b$ are deterministic. As long as the spreads are deterministic, we can apply
the theorems to interest rate derivatives.

\begin{example}[A zero coupon bond] \label{example:treasury} Let us consider
  a hedger purchasing a  bond with unit notional amount, i.e.,
  $\frD = -\1_{\llbracket T, \infty\llbracket }$. We assume that for
  $i \in \{H, C\}$, $\df G\IT = -h\IT G\IT \df t$, where $(h^i_t)_{t \geq0}$ are
  deterministic processes. To consider a government counterparty, one may want
  to set $h^C=0$. We assume that the OIS rate $r$ is given by the following:
\begin{align}
  \df r_t = \kappa(\mu-r_t)\df t + \zeta\df W_t, \label{ex:model:r}
\end{align}
for some $\kappa , \mu , \zeta >0$. Therefore, we have:
\begin{align}
  \sigma^1 =& \sigma^H = \sigma^C=-\frac{\zeta [1-e^{-\kappa(T-t)}]}{\kappa}.\nonumber 
\end{align} 
We also assume $\rho =\emptyset$. Then $\sum_{i \in I}\pit^i\sigma^i = \sigma^1(\sum_{i \in I}\pit^i)$. Consequently,
$  \phi(z) = \alpha(z+Z^P) = (\sigma^1)^{-1}(z+Z^P)$. 
From \cref{ex:model:r}, we can find:
\begin{align}
 r_t = r_0e^{-\kappa t} + \mu(1-e^{-\kappa t}) + \zeta \int_{0}^{t}e^{-\kappa(t-u)}\df W_u.\nonumber  
\end{align}
Then by Corollary 3.19 in \cite{di2009malliavin}, 
$D_\theta r_t = \zeta e^{-\kappa(t-\theta)}$ for any $\theta \leq t$. It then follows that:
\begin{align}
  D_\theta B^{-1}_t =& -B_t^{-1}\int_{\theta}^{t}D_\theta r_s \df s = -B_t^{-1}\zeta \int_{\theta}^{t}e^{-\kappa(s-\theta)} \df s
  =-B_t^{-1}\frac{\zeta [1-e^{-\kappa(t-\theta)}]}{\kappa}.\nonumber 
\end{align}
Recalling that $\xi =  -B^{-1}_T$:
\begin{align}
  \xi - \alpha_\theta D_\theta \xi =&  \xi - (\sigma^1_\theta)^{-1} D_\theta \xi\nonumber \\
  =&-B_T^{-1}+(\sigma^1_\theta)^{-1}B_T^{-1}\frac{\zeta [1-e^{-\kappa(T-\theta)}]}{\kappa}=0.\nonumber 
\end{align}
It then follows that for any $\theta \leq t$, $\pt_t- \alpha_\theta D_\theta\pt_t=0$. Moreover,
  $\alpha_t^\top Z^P_t =  (\sigma^1_t)^{-1} D_t\pt_t  = -\dsE[-B_T^{-1}\mid\calF_t]<0$. Therefore,
  \cref{thm:czf2} and \cref{thm:cxl} are satisfied.
  Finally to satisfy \cref{thm:cx}, we assume $\epsilon \geq0$. 
Then by (\rom{2}) in \Cref{thm:clean}, $(Y^\dsF, Z^\dsF)$ satisfies:
\begin{align}
  Y^\dsF_t =& \int_{t}^{T} -\bigg[\Big(Y^\dsF_s
              +L^m\pt_s+\Bt^\epsilon_s-\phi_s(Z^\dsF_s)\Big)s^{\ell}_s
              -s^\epsilon_s\Bt^\epsilon_s+h_sY^\dsF_s\bigg]\df s\nonumber\\
  &+\int_{t}^{T}
    -\big[\1_{\pt_s\geq0}h^H_sL^HL^m+\1_{\pt_s<0}h^C_sL^CL^m\big]\pt_s\df s
    -\int_{t}^{T}(Z^\dsF_s)^\top\df W_s, \label{ex1:yf}
\end{align}
and we have FCA$=0$. \qed 
\end{example}
\begin{remark}
 It is worth noting the financial meaning of $\epsilon \geq0$ in
 \Cref{example:treasury}. Recall that $\epsilon$ represents the initial funding state
 of the bank. When $\epsilon <0$, the funding benefit from the new contract is used for
 deducing the bank's funding cost, rather than increasing the funding benefit, i.e.,
 FCA may still exist.
\end{remark}
To find an analytical form of (\ref{ex1:yf}), let
$V^\dsF \coloneqq Y^\dsF + \pt$, $\Pi^\dsF = Z^\dsF + Z^P$, and let $\dsQ^\ell$
denote an equivalent measure such that $(B^\ell)^{-1}S^i$ is $(\dsQ^\ell, \dsG)$-local
martingales. In particular, $\{W^\ell_t\}_{t\geq 0 } \coloneqq \big\{W_t
-\int_{0}^{t}\alpha_ss^\ell_s\df s \big\}_{t\geq 0 }$
is an $(\dsF, \dsQ^\ell)$-Brownian motion. Then, (\ref{ex1:yf}) becomes:
\begin{align}
  V^\dsF_t =&\xi
  + \int_{t}^{T}\Big[-(s^\ell_s+h_s)V^\dsF_s +(s^\epsilon_s-s^\ell_s)\Bt^\epsilon_s+\beta_s \pt_s\Big]\df s
  -\int_{t}^{T}\Pi^\dsF_s \df W^\ell_s,\nonumber 
\end{align}
where $\beta_t\coloneqq (1-L^m)s^\ell_t + h_t-h^C_tL^CL^m$.
Let $A_t \coloneqq \exp{\big[-\int_{0}^{t}(s^\ell_s+h_s) \df s\big]}$. Then $V^\dsF$
can be represented by the following conditional expectation form:
\begin{align}
  V^\dsF_t = A^{-1}_t\dsE^\ell\bigg[A_T\xi +
  \int_{t}^{T}A_s\big[(s^\epsilon_s-s^\ell_s)\Bt^\epsilon_s+\beta_s \pt_s\big]
  \df s\bigg|\calF_t\bigg], \label{analyticform}
\end{align}
where $\dsE^\ell$ is the expectation under $\dsQ^\ell$.  However, there is no
additional advantage
with the \textit{clean price} convention and we cannot find a closed-form
solution of $V^\dsF$. This is due to the mismatch of the pricing measures in
$\dsE^\ell[\beta_s \pt_s\big|\calF_t] = \dsE^\ell[ \dsE[\beta_s\xi|\calF_s]\big|\calF_t]$.
To avoid this difficulty, \cite{brigo2017funding} considered un-collateralized
contracts with null cash-flow at defaults. On the other hand,
\cite{bichuch2017arbitrage} assumed that the close-out amount and collateral are
calculated by the risk-neutral price under $\dsQ^\ell$ (or $\dsQ^b$), namely:
\begin{align}
  e_t =& (B^\ell_t)^{-1}\dsE^{\ell}[(B^\ell_T)^{-1}\Delta \frD_T|\calF_t], \label{fvaprice}\\
  m_t =& (1-L^m)e_t. 
\end{align}
In these cases, the pricing measures are aligned and a closed-form solution can
be obtained. However, note that (\ref{fvaprice}) is
  \textit{clean price} \texttt{+} "the hedger's FVA".
As will be shown later, when \textit{replacement cost} is assumed, the
inconsistency of pricing measures does not appear and closed-form solutions are
given. However, recall that the hedger's funding information is already
considered in $V_-$. Therefore, the consistency of
pricing measures is inherent  in \textit{replacement cost}.

\subsection{Replacement Cost}
In the next example, we deal with a
non-Markovian case. This is one advantage of BSDEs and Malliavin calculus.
\begin{example}[Floating strike geometric Asian call option with
  \textit{replacement cost}]
Let $n=1$, $\rho = \{H, C\}$, $e = V_- - B^\epsilon$, and the traded assets are
given by the following:
\begin{align}
  \df S^1_t &= rS^1_t\df t + \sigma^1_tS^1_t\df W_t,\nonumber\\
  \df S^i_t &= rS^1_t\df t - S_{t-}^i\df M^i_t, ~~i \in \{H, C\}.\nonumber
\end{align}
Since $\rho = \{H, C\}$, $\sum_{i\in I}\pit^i\sigma^i =\sigma^1 \pit^1$ and
$\sum_{i\in I\setminus\rho}\pit^i = \pit^1$. Therefore:
\begin{align}
  \phi_t(z) = \alpha_t(z+Z^P), ~~~\text{and}~~~\alpha_t = (\sigma^1_t)^{-1}.
\end{align}
We also consider a floating strike Asian call option:
\begin{align}
  \frD = \1_{\llbracket T, \infty\llbracket}B_T(\St^1_T - B_T^{-1}K \dsI_T)^+,
  ~~~~\dsI_T \coloneqq \exp{\bigg(\frac{1}{T}\int_{0}^{T}\ln{(S^1_u)}\df u\bigg)}. 
  \nonumber 
\end{align}
We will then show (\ref{thm2:cond}). By
Theorem 3.5 in \cite{di2009malliavin}, for $\theta \leq T$:
\begin{align}
  \alpha_\theta D_\theta \dsI_T = &  \frac{\alpha_\theta}{T}\dsI_T\int_{\theta}^{T}D_\theta\ln{(S^1_u)}\df u
  =\frac{\alpha_\theta}{T}\dsI_T\int_{\theta}^{T}(S^1_u)^{-1}\sigma^1_\theta S^1_u\df u =\frac{T-\theta}{T}\dsI_T.\nonumber 
\end{align}
Since $  \alpha_\theta D_\theta\St^1_T = \St^1_T$, we have:
\begin{align}
 L^m\xi -\alpha_\theta D_\theta \xi =& \1_{\St^1_T \geq B_T^{-1}K \dsI_T}\big[L^m(\St^1_T - B_T^{-1}K
                    \dsI_T)-\alpha_\theta D_\theta(\St^1_T - B_T^{-1}K  \dsI_T)\big]\nonumber \\
  =&\1_{S^1_T \geq K \dsI_T}\Big[(L^m-1)\St^1_T - B^{-1}_TK
     \dsI_T\Big(L^m-\frac{T-\theta}{T}\Big)\Big]\nonumber \\
  \leq& (L^m-1)(\St^1_T- B^{-1}_TK \dsI_T)^+\leq0. \nonumber 
\end{align}
Therefore, FBA$=0$, where $\epsilon \leq0$. 
\end{example}
As the final example, we deal with a bond option. 
\begin{example}[A bond option with \textit{replacement
    cost}]\label{example:bondoption}
  We assume that the same market conditions hold
  as in \Cref{example:treasury}. Let $\frD_t = \1_{T\leq t}(S^1_{T, U}-K)^+$,
  where $K > 0$ and $S^1_{\cdot, U}$ is a zero coupon bond with $U>T$ as its
  maturity. We consider two defaultable bonds with the same maturity, i.e.,
 $S^1_{t, U} \coloneqq B_t\dsE[B_U^{-1}\mid\calF_t]$ and 
$ S^i_{t, U} \coloneqq B_t\dsE[\1_{U < \tau^i}B_U^{-1}\mid\calG_t], ~i \in \{H, C\}$. 
Recall  that $  \phi(z) = \alpha(z+Z^P) = (\sigma^1)^{-1}(z+Z^P)$, and 
$  \sigma^1 = \sigma^H = \sigma^C-\zeta [1-e^{-\kappa(U-t)}]/\kappa$. 
In addition, recall the definition of $\xi$: 
$ \xi =  \1_{ S^1_{T, U} \geq K}(\St^1_{T, U}-KB^{-1}_T)$. 
We can see that:
\begin{align}
  \alpha_\theta D_\theta (KB^{-1}_T)
  =(\sigma_\theta^1)^{-1}K D_\theta B_T^{-1} = -KB_T^{-1}(\sigma_\theta^1)^{-1}\frac{\zeta
     [1-e^{-\kappa(T-\theta)}]}{\kappa}
 = KB_T^{-1} \frac{1-e^{-\kappa(T-\theta)}}{1-e^{-\kappa(U-\theta)}}.\nonumber 
\end{align}
Furthermore, $ \alpha_\theta D_\theta \St^1_{T, U} = (\sigma_\theta^1)^{-1}\sigma_\theta^1  \St^1_{T, U} = \St^1_{T, U}$. 
It then follows that, on $\{ S^1_{T, U} \geq K\}$:
\begin{align}
  L^m \xi - \alpha_\theta D_\theta \xi
  =&L^m(\St^1_{T, U}-KB^{-1}_T)- \alpha_\theta D_\theta(\St^1_{T, U}-KB^{-1}_T)\nonumber\\
  =&(L^m-1)\St^1_{T, U} - KB^{-1}_T\bigg(L^m -
     \frac{1-e^{-\kappa(T-\theta)}}{1-e^{-\kappa(U-\theta)}}\bigg)\nonumber\\
  \leq&(L^m-1)\St^1_{T, U} - KB^{-1}_T(L^m -1)
     =(L^m-1)(\St^1_{T, U}- KB^{-1}_T) \leq0.\label{direction}
\end{align}
Therefore, by (\rom{1}) in \Cref{thm:replacement},  $(Y^\dsF,
Z^\dsF)$ is given by:
\begin{align}
  Y^\dsF_t = \int_{t}^{T} -\bigg[\Big(L^mY^\dsF_s
              +L^m\pt_s-\phi_s(Z^\dsF_s)\Big)s^{b}_s
              +h^H_sL^HL^m(Y^\dsF_s+\pt_s)\bigg]\df s
  -\int_{t}^{T}Z^\dsF_s\df W_s, \label{ex3:bsde}
\end{align}
and FBA$=0$. Note that  $\text{DVA}\not=0$ because $\xi \geq0$, but
$\text{FBA}=0$ when $\epsilon \leq0$ because $\xi$ increases in $S^1_{t, U}$. \qed
\end{example}
As we have shown, FBA and DVA are affected by different mathematical structures
of the pay-off $\xi$. We have summarized the cases of vanilla options in
\Cref{table:FBA.DVA}. Only when the hedger sells a put option are FBA and DVA 
both positive, but still $\text{FBA}\not=\text{DVA}$.
\begin{table}[htp]
  \centering
\begin{tabular}{ |c | c | c |}
    \hline
     & FBA  &  DVA \\ \hline
    buy/call & positive &  nil\\ \hline
  sell/call & nil &  positive\\ \hline
    buy/put & nil&  nil \\ \hline
    sell/put &  positive &  positive \\ \hline
\end{tabular}
\caption{DVA and FBA with respect to option contracts}  
\label{table:FBA.DVA}     
\end{table}

\subsection{A Closed-form Solution of a Call Option under Replacement Cost}
Under \textit{replacement cost}, we can find a closed-form solution.
As an example, we discuss  the solution of a stock call option. 
Let $n=1$, $e=V_- -B^\epsilon$, $\rho = \{H, C\}$,
$\epsilon\leq0$, $\frD = \1_{\llbracket T, \infty\llbracket}(S^1_T - K)^+$, where:
\begin{align}
 \df S^1_t =& r S^1_t \df t + \sigma^1 S^1_t \df W_t, \nonumber \\  
 \df S^i_t =& r S^i_t \df t +  S^i_{t-} \df M^i_t, ~i\in \{H, C\},    \nonumber 
\end{align}
for some constants $r$ and $\sigma^1$. We also assume that $s^b$, $h^H$ are
constant. Recall the following: 
\begin{align}
  \xi = B_T^{-1}(S^1_T - K)^+, ~~~~ \phi_t(z) = \alpha(z+Z^P_t) =(\sigma^1)^{-1}(z+Z^P_t).\nonumber 
\end{align}
It is easy to verify that:
\begin{align}
  L^m \xi -\alpha_\theta D_\theta\xi
  =   L^m \xi -(\sigma^1)^{-1}D_\theta\xi =\1_{S^1_T \geq K}\big[(L^m-1)\St^1_T - KL^mB_T^{-1}
     \big]\leq0. \nonumber 
\end{align}
Therefore, by (\rom{1}) in \Cref{thm:replacement}, FBA$=0$ and $(Y^\dsF,
Z^\dsF)$ satisfies:
\begin{align}
  Y^\dsF_t = \int_{t}^{T} -\bigg[\Big(L^mY^\dsF_s +L^m\pt_s-\phi_s(Z^\dsF_s)\Big)s^{b}_s\bigg]\df s 
  +\int_{t}^{T}
    \big[-h^H_sL^HL^m(Y^\dsF_s+\pt_s)\big]\df s
    -\int_{t}^{T}Z^\dsF_s\df W_s. \label{closed:bsde}
\end{align}
Let $V^\dsF \coloneqq Y^\dsF + \pt$, $\Pi^\dsF \coloneqq Z^\dsF + Z^P$. In addition, let
$\dsQ^b$ denote an equivalent measure, under which:
\begin{align}
    W^b_t \coloneqq &W_t -\int_{0}^{t}s^b\alpha_s  \df s \label{bB}
\end{align}
is an $(\dsF, \dsQ^b)$-Brownian motion and $\dsE^b$ denotes the expectation under
$\dsQ^b$. Then (\ref{closed:bsde}) becomes:
\begin{align}
 V^\dsF_t = \xi -\int_{t}^{T} (s^b + h^HL^H)L^m V^\dsF_s - \int_{t}^{T}\Pi^\dsF_s \df
  W^b_s.  \nonumber 
\end{align}
It then follows that:
\begin{align}
  V^\dsF_t = \chi^{-1}_t\dsE^b[A_T\xi |\calF_t], ~~~\chi_t \coloneqq e^{-\beta t},
              ~~~\beta\coloneqq (s^b + h^HL^H)L^m. \label{firstrep}
\end{align}
Note that there is no inconsistency of pricing measures. This is because we
considered the hedger's funding cost and benefit in the close-out amount, 
$e=V_-$. To represent (\ref{firstrep}) in an explicit form, we write:
  $\xi = (B^b_T)^{-1}B^b_TB^{-1}_TB_T\xi =e^{s^b T}(B^b_T)^{-1}B_T\xi$. 
Then (\ref{firstrep}) becomes, for $t < \btau$:
\begin{align}
 V^\dsF_t = \chi_t^{-1}\chi_Te^{s^b
  T}(B^b_t)^{-1}B_t^b\dsE^b[(B^b_T)^{-1}B_T\xi|\calF_t].\nonumber 
\end{align}
More explicitly:
\begin{align}
  V_t=B_tV^\dsF_t   =&\exp{\big[ \big(s^b(1-L^m)-h^HL^HL^m  \big)(T-t)\big]}\calC^b(t, S^1_t),
  \label{closed:sol}\\ 
  \calC^b(t, S^1_t)\coloneqq& S^1_t\Phi(d(t, S^1_t))
                              -Ke^{-R^b(T-t)}\Phi\big(d(t, S^1_t)-\sigma^1\sqrt{T-t}\big),\\
  \Phi(x)\coloneqq&\int_{-\infty}^{x}e^{-y^2/2}/\sqrt{2\pi} \df y,\\
  d(t, x) \coloneqq& \frac{\ln{(S^1_t/K)} + \big(R^b +(\sigma^1)^2/2\big)(T-t)}{\sigma^1\sqrt{T-t}}.
\end{align}
Note that $\partial_{s^\ell}V_t^\dsF = 0$. Furthermore, in (\ref{closed:sol}), we can see that
  $(1-L^m)s^b+L^m(-h^HL^H)$ 
is a weighted sum of $s^b$ and $-h^HL^H$. This shows how the effect of DVA is
transferred to FCA as $L^m$ changes. As $L^m$ increases, the effect of funding
cost is diminished, since the cost of posting the collateral is more expensive than
the interest rate remunerating on the collateral, i.e.,  the OIS rate. More
precisely, 
  $\exp{(-h^HL^HL^m(T-t))}$ 
is a deduction from  DVA, while
  $\exp{(s^b(1-L^m)(T-t))}$
is a compensation for the hedger for  posting the collateral.  However, this is not
the only part of FCA. The other part of FCA is the cost to acquire $S^1$, and it is
included in $\calC^b$.

\section{Conclusion}
In summary, we discussed the binary nature of FVA. According to the binary
nature of FVA, we can recover linear BSDEs and analytical solutions can be
found. Under \textit{replacement cost}, the analytical solution can be
represented in a closed-form. As a byproduct, this feature of FVA explains how
FBA and DVA are different.  In addition, this result provides an interpretation
for banks purchasing Treasury bonds with the presence of funding rates that are
higher than the OIS rate.

\bibliography{binary.funding}   
 
\appendix

\section{Spaces of Random Variables and Stochastic Processes}
\label{sec:spaces}       
\begin{definition} Let $m \in \mathbb{N}$ and $p \geq 2$.
\begin{itemize}
\item $\dsL^p_T$: the set of all $\calF_T$-measurable random variables $\xi$, such that:
 $ \Vert\xi \Vert_p\coloneqq \dsE[|\xi|^p]^{\frac{1}{p}}<\infty$. 
\item $\dsS^p_T$: the set of all real valued,
  $\dsF$-adapted, c\`adl\`ag\footnote{Right continuous
    and left limit} processes $(U_t)_{t\geq0}$, such that:
  \begin{align}
    \Vert U \Vert_{\dsS_T^p} \coloneqq \dsE\big(\sup\limits_{t \leq T}\vert
    U_t\vert^p\big)^{\frac{1}{p}}< \infty.  \nonumber
  \end{align}
\item $\dsH^{p, m}_T$: the set of all $\dsR^m$-valued, $\dsF$-predictable
   processes $(U_t)_{t\geq0}$, such that:
  \begin{align}
    \Vert U \Vert_{\dsH^p_T} \coloneqq \dsE\Big(\int_0^T\big\vert U_t \big\vert^p\df
    t\Big)^{\frac{1}{p}} < \infty. \nonumber 
  \end{align}
\item $\dsH^{p, m}_{T, loc}$ : the set of all $\dsR^m$-valued, $\dsF$-predictable
   processes $(U_t)_{t\geq0}$, such that:
 $\int_0^T\big\vert U_t \big\vert^p\df t< \infty, \text{a.s.}$ 
\end{itemize}
\end{definition}
\section{Reduction of Filtration}
\label{app:red}
The next lemma is from \cite{bielecki2008pricing} and Chapter 5 in
\cite{bielecki2010credit}.  
\begin{lemma}\label{lem:red} Let $i \in \{H, C\}$.
\begin{enumerate}[(i)]
\item Let $U$ be an $\calF_s$-measurable, integrable random variable for some
  $s \geq0$. Then, for any $t \leq s$:
  \begin{align}
    \dsE(\1_{s <\tau}U\vert \calG_t) =& \1_{t < \tau}G_t^{-1}
    \dsE (G_sU \vert\calF_t), \nonumber\\
    \dsE(\1_{s <\tau^i}U\vert \calG_t) =& \1_{t < \tau^i}(G^i_t)^{-1}
    \dsE (G^i_sU \vert\calF_t). \nonumber
  \end{align}
\item Let $(U_t)_{t \geq0}$ be a real-valued, $\dsF$-predictable process and
  $\dsE\vert U_{\btau} \vert < \infty$. Then:
  \begin{align}
    \dsE(\1_{\tau = \tau^i \leq T}U_{\tau}\vert \calG_t) =
    \1_{t< \tau}G_t^{-1}\dsE\Big( \int_{t}^{T}h^i_s G_sU_s\df s \Big\vert \calF_t\Big). \nonumber
  \end{align}
\end{enumerate}
\end{lemma} 

\section{Risk-Neutral Probability}
\label{sec:risk-neutral}
In this section, we provide a brief justification of the unique risk-neutral
probability. Let us assume that the underlying assets are given by:
\begin{align}\label{app:sde1}
  \begin{cases}
    \df S^i_t = \mu^i_tS_t^i\df t + \sigma^i_tS^i_t \df W^\dsP_t,  &i \in \{1, \cdots, n\},\\
    \df S^i_t = \mu^i_tS_t^i\df t + \sigma^i_tS^i_t \df W^\dsP_t-S^i_{t-}\df M^{\dsP, i}_t,
    &i \in \{H, C\},
  \end{cases}
\end{align}
where $\dsP$ is the real world probability, $(\mu^i_t)_{t\geq0}$ are
$\dsF$-progressively measurable processes, $(W^\dsP_t)_{t\geq0}$ is a Brownian motion
under $\dsP$, $(M^{\dsP, i}_t)_{t\geq0}$ are given by
$M^{\dsP, i}_t \coloneqq \1_{t\wedge \tau \geq \tau^i} - \int_{0}^{t \wedge \tau}h^{\dsP, i}_s\df s $,
for some $\dsF$-progressively measurable processes $h^{\dsP, i}$. Denote
$\mu \coloneqq (\mu^1, \cdots, \mu^n)^\top$. Since $\Sigma$ is non-singular, there exists a unique
process $\Lambda\in \dsR^n$, such that $\Sigma\Lambda = \mu - r\1$. Then, we can put:
\begin{align}\label{app:sde2}
  \begin{cases}
    \df S^i_t = r_tS_t^i\df t + \sigma^i_t S^i_t(\df W^\dsP_t+\Lambda_t \df t),  &i \in \{1, \cdots, n\},\\
    \df S^i_t = r_tS^i_t\df t
    +\sigma^i_tS^i_t (\df W^\dsP_t + \Lambda_t\df t)+[(\mu^i_t-r_t)-\sigma^i_t\Lambda_t]S_t^i\df t  -S^i_{t-}\df M^{\dsP, i}_t,
    &i \in \{H, C\},
  \end{cases}
\end{align}
Now, we define:
\begin{align} 
  W^{\dsQ}_t &\coloneqq W^{\dsP}_t + \int_{0}^{t}\Lambda_s \df s,
  \nonumber\\
  M^{i, \dsQ}_t &\coloneqq \1_{t \wedge \tau \geq \tau^i} + \int_{0}^{t \wedge \tau}h^i_s \df s,~~i \in \{H, C\},\nonumber\\
  h^{i, \dsQ}_t &\coloneqq h^{\dsP, i}_t+(\mu^i_t-r_t)-\sigma^i_t\Lambda_t.
\end{align}
Therefore, so that $B^{-1}S^i$, $i \in \{1, \cdots, n\}$, are martingales under some
probability $\dsQ$, $(W^{\dsQ}, M^{H, \dsQ}, M^{C, \dsQ})$ should be martingales. The choice of
Radon-Nykod\'ym density process can be found by slight generalization of
Proposition 5.3.1 in \cite{bielecki2010credit}. Moreover, note that the choice
of measure is unique by the uniqueness of $\Lambda$.

\section{Incremental Funding Impacts}
\label{sec:inc}
\begin{example} 
Let us consider $n=1$, $\rho =\emptyset$,
$\frD = \1_{\llbracket T, \infty \llbracket} \xi B_T$ for some
$\xi \in \dsL^2_T$. We ignore default risks and set $\pi^H = \pi^C = m=0$. In this
case, $(V, \pi^1)$ is given by:
\begin{align}
  V_t = \xi B_T +B_T^\epsilon - \int_{t}^{T}\Big[(V_s - \pi^1_s)^+R^\ell_s
  -(V_s - \pi^1_s)^-R^b_s + r_s \pi^1_s\Big]\df s  -\int_{t}^{T} \pi^1_s \sigma^1_s \df W_s. \nonumber
\end{align}
To consider a net profit/loss to the hedger without the \textit{legacy portfolio}, define:
 \begin{align}
   v\coloneqq B^{-1}(V -B^\epsilon), \label{ex.net}
 \end{align}
 and we denote that $\pit^1 \coloneqq B^{-1}\pi^1$,
 $\Bt^\epsilon = B^{-1}B^\epsilon$. Then $(v, \pit^1)$ is given by the following:
\begin{align}
  v_t = \xi - \int_{t}^{T}\Big[(v_s +\Bt^\epsilon_s- \pit^1_s)^+s^\ell_s
  -(v_s +\Bt^\epsilon_s - \pit^1_s)^-s^b_s -s^\epsilon_s\Bt^\epsilon_s\Big]\df s  -\int_{t}^{T} \pit^1_s
  \sigma^1_s \df W_s. \label{ifva.bsde}
\end{align}
Assuming there exists a unique solution $(v, \pit^1) \in \dsS^2_T \times \dsH^2_T$ of
(\ref{ifva.bsde}), then for any $t \geq0$, we define:
\begin{align}
\FBA^ \Delta_t \coloneqq &\dsE\bigg[
\int_{t}^{T}\big[(v_s +\Bt^\epsilon_s- \pit^1_s)^+s^\ell_s-(s^\epsilon_s\Bt^\epsilon_s)^+\big]\df s\bigg\vert \calF_t\bigg], \label{ex.fba.d}\\
 \text{FCA}^\Delta_t \coloneqq &\dsE\bigg[
\int_{t}^{T}\big[ (v_s +\Bt^\epsilon_s- \pit^1_s)^-s^b_s-(s^\epsilon_s\Bt^\epsilon_s)^- \big] \df s\bigg\vert \calF_t\bigg].  \label{ex.fca.d} 
\end{align}
$\text{FBA}^\Delta$ (resp. $ \text{FCA}^\Delta$) represents the incremental funding
benefits (resp. costs) for entering the new contract. $\FBA$ and $\FCA$ stand
for \textit{funding benefit adjustment} and \textit{funding cost adjustment},
respectively. The combination of $\FBA$ and $\FCA$ is often called
\textit{funding value adjustment} (FVA). More precisely:
\begin{align}
  \FVA \coloneqq \FCA - \FBA \nonumber
\end{align}

Notice that, in (\ref{ex.fba.d})-(\ref{ex.fca.d}), as $\epsilon$ increases
$v +\Bt^\epsilon -\pi^1$ is more likely to positive. Consider a case in which 
$v -\pi^1$ is negative, but $v +\Bt^\epsilon -\pi^1$ is non-negative. If we ignore the
incremental effect, i.e., $\epsilon=0$, the dealer should consider the increased
funding cost. However, in the view of incremental effects, instead, the
deduction of funding benefit should be included in the derivative
price. The dealer also needs to consider the opportunity cost for not
entering the new contract, e.g., if $\epsilon \geq0$, the lost (discounted) benefit at
$t \leq T$ would be:
\begin{align}
  \Bt^\epsilon_T -\Bt^\epsilon_t=\int_{t}^{T} \df \Bt^\epsilon_s  =\int_{t}^{T} s^\epsilon_s\Bt^\epsilon_s\df s=\int_{t}^{T} (s^\epsilon_s\Bt^\epsilon_s)^+\df s.
\end{align}
The difference between the
two impacts is the actual net benefit and cost, $\text{FBA}^\Delta$ and
$\Delta \text{FCA}^\Delta$, that should be charged to the counterparty. Indeed, by
(\ref{i.price}), (\ref{ex.net}) and (\ref{ifva.bsde}), the dealer would want to charge
 $ p = \dsE[\xi] - \FBAd_0 +  \FCAd_0$. 
In addition, if $R^\ell=R^b$, (\ref{ifva.bsde}) becomes:
\begin{align}
  v_t = \xi - \int_{t}^{T}(v_s - \pit^1_s)s^\ell_s
  \df s  -\int_{t}^{T} \pit^1_s  \sigma^1_s \df W_s. 
\end{align}
Thus, under linear funding models that $s^\ell = s^b$, $v$ does not depend on
$B^\epsilon$.  \qed
\end{example}
\section{The Proofs of Lemmas} 
\label{app:lemmas}
\begin{proof}[Proof of \Cref{lem:pf}]
  (\rom{1}) is from the definition and (\rom{2}) is a directly obtained from
  (\rom{1}).  For (\rom{3}), notice that
  $B^{-1}P + \int_{0}^{\cdot}B^{-1}_s\df \frD_s$ is an $(\dsF, \dsQ)$-local
  martingale. Thus, by (local) martingale representation property, there exists
  $Z^P \in \dsH^{2, n}_{T, loc}$, such that:
$B_t^{-1}P_t + \int_{0}^{t}B^{-1}_s\df \frD_s = \int_{0}^{t}(Z^P_s)^{\top}\df W^{}_s$
for any $t \geq0$.
Therefore, $P_t$ satisfies the SDE:
  $\df P_t = r_tP_t\df t + B_t(Z^P_t)^{\top}\df W^{}_t - \df \frD_t$.
By (\rom{3}), $P$ is an $\dsF$-adapted c\`adl\`ag process, but $\tau$ avoids any
$\dsF$-stopping time. Thus, $\Delta P_{\tau}=0$ almost surely, and equivalently $P_{\tau-}
= P_{\tau}$ a.s.        
\end{proof}
\begin{proof}[Proof of \Cref{lemma:reduction}]
By It\^o's formula:
\begin{align}
  \df (\1_{t <\btau}Y^{\dsF}_t) = \1_{t \leq \btau}\df
  Y^\dsF_t - \bm{\delta}_{\btau}(\df t)Y^\dsF_{\btau}
  =&
  \1_{t \leq \btau}\df Y^\dsF_t - \1_{t \leq \btau}Y^\dsF_t\df M_t - \1_{t \leq
     \btau}h^0_tY^\dsF\df t, \nonumber\\
  =&\1_{t \leq \btau}\df Y^\dsF_t - \1_{t \leq \btau}Y^\dsF_t\df M^H_t
     - \1_{t \leq \btau}Y^\dsF_t\df M^C_t- \1_{t \leq
     \btau}h_tY^\dsF\df t, \nonumber 
\end{align}
and:
\begin{align}
\df\Big(\1_{t \geq \btau} \big[-\1_{\btau=\tau=\tau^H}
  \Jt^{ H}_{\tau}(Y^{\dsF}_{\tau-})+&\1_{\btau=\tau=\tau^C}\Jt^{  C}_{\tau}(Y^{\dsF}_{\tau-}) \big]\Big) \nonumber\\
  =&-\1_{\tau=\tau^H\leq T }\bm{\delta}_{\tau}(\df  t)\Jt^{ H}_t(Y^{\dsF}_{t-})
  +\1_{\tau=\tau^C\leq T }\bm{\delta}_{\tau}(\df t)\Jt^{  C}_t(Y^{\dsF}_{t-}) \nonumber\\
  =&-\1_{t\leq \btau}\Jt^{ H}_t(Y^{\dsF}_{t-})\df M^H_t - \1_{t \leq
     \btau}h_t^H\Jt^{ H}_t(Y^{\dsF}_{t-})\df t\nonumber\\
   &+\1_{t\leq \btau}\Jt^{  C}_t(Y^{\dsF}_{t-})\df M^C_t
     + \1_{t \leq \btau}h_t^C\Jt^{  C}_t(Y^{\dsF}_{t-})\df t. \nonumber 
\end{align}
If we take $Y^\dsG$ as in (\ref{fsol1}):
\begin{align}
 \df Y_t^\dsG =& \1_{t \leq \btau}\df Y^\dsF_t - \1_{t \leq \btau}Y^\dsF_t\df M^H_t
     - \1_{t \leq \btau}Y^\dsF_t\df M^C_t- \1_{t \leq
  \btau}h_tY^\dsF\df t \nonumber\\
  &-\1_{t\leq \btau}\Jt^{ H}_t(Y^{\dsF}_{t-})\df M^H_t - \1_{t \leq
     \btau}h_t^H\Jt^{ H}_t(Y^{\dsF}_{t-})\df t
    \nonumber\\
  &+\1_{t\leq \btau}\Jt^{  C}_t(Y^{\dsF}_{t-})\df M^C_t
    + \1_{t \leq \btau}h_t^C\Jt^{  C}_t(Y^{\dsF}_{t-})\df t \nonumber\\
  =&\1_{t \leq \btau}\Big[-g^\dsF_t(Y^\dsF_t, Z^\dsF_t) - h_tY^\dsF_t -
     h^H_t\Jt^{ H}_t(Y^\dsF_t) + h^C_t\Jt^{  C}_t(Y^\dsF_t)\Big]\df t + \1_{t \leq
     \btau}(Z^\dsF_t)^\top \df W_t \nonumber\\
             & -\1_{t \leq \btau}[Y^\dsF_t + \Jt^{ H}_t(Y^\dsF_t)]\df M^H_t
               -\1_{t \leq \btau}[Y^\dsF_t - \Jt^{  C}_t(Y^\dsF_t)]\df M^H_t. \nonumber 
\end{align}
Therefore, by (\ref{fgenerator}), (\ref{fsol1})-(\ref{fsol3}) give
 a solution for (\ref{GBSDE}). Moreover, if (H)-hypothesis holds,
the (unique) martingale representation property holds by $W$ and $M^i$,
$i \in \{H, C\}$. Therefore, by Theorem 4.1 in \cite{crepey2015bsdes}, if
$(Y^\dsG, Z^\dsG)$ solves (\ref{GBSDE}), $((Y^\dsG)^{\tau-}, Z^\dsG\1_{\llbracket 0, \tau\llbracket})$ solves
(\ref{FBSDE}) as well.
\end{proof}

\section{The Idea of the Main Theorem}
\label{sec:idea}
\begin{example}[Stock forward contract with \textit{clean price}]\label{example:idea}
  Consider $n=1$, $e=P$. For simplicity, we assume that all parameters are
  constant and let the traded assets
  $(S^1, S^H, S^C)$ be given by:
\begin{align}
\df S^1_t =& r_tS^1_t \df t + \sigma^1S^1_t\df W_t, \nonumber \\
  \df S^i_t =& r_tS^i_t \df t -S^i_{t-}\df M^i_t, ~~i \in \{H, C\}. \nonumber 
\end{align}
Moreover, we assume that the defaultable bonds can be traded through repo markets,
i.e., $\rho =\{H, C\}$. Then $\sum_{i \in I}\pit^i \sigma^i = \pit^1\sigma^1$ and $\sum_{i \in
  I\setminus\rho}\pit^i = \pit^1$, and therefore
 $\phi(z) = \alpha(z+Z^P) = (\sigma^1)^{-1}(z+Z^P)$. 
Let $\frD = \1_{\llbracket T, \infty\llbracket}(S^1_T - K)$, for some
$K\geq0$. We denote $\St^1 \coloneqq B^{-1}S^1$. 
Recall the definition $\pt = B^{-1}P$, 
$\Jt^{  H}  = L^HL^m  \pt^+$, and 
 $\Jt^{  C}  = L^CL^m  \pt^-$.
Thus, the generator $g^\dsF$ becomes:
\begin{align}
  g^\dsF_t(y, z) =
  & -\big[y +L^m\pt_t+\Bt^\epsilon_t-\alpha(z+Z^P_t)\big]^+s^{\ell}
    + \big[y +L^m\pt_t+\Bt^\epsilon_t-\alpha(z+Z^P_t)\big]^-s^{b}+s^\epsilon\Bt^\epsilon_t\nonumber\\
  &-h^HL^HL^m \pt^+_t
    +h^CL^CL^m \pt^-_t -hy. \nonumber 
\end{align}
To explain the idea, the hedger should pay $S^1_T-K$ at maturity or
$P_t = S_t^1 -B_tB_T^{-1}K$ at an early termination $t < T$. For the payment, she
needs to retain $S^1$. To purchase $S^1$, the hedger may need to borrow money, so it
is expected that $s^\ell$ does not play an important role in maintaining the
hedging portfolio. Therefore, we suppose:
\begin{align}
  Y^\dsF +L^m\pt- \alpha(Z^\dsF +Z^P) \leq 0, ~~\dqdt. \label{idea:ineqf}
\end{align}
If the effect of \cref{idea:ineqf} is not dominated by the \textit{legacy
  portfolio}, we can recover a linear BSDE. 

To this end, we consider $(Y^\#, Z^\#)$ satisfying:
\begin{align}
  Y^\#_t =& -\int_{t}^{T} \bigg[\big(Y^\#_s +L^m\pt_s +\Bt^\epsilon_s
            -\alpha(Z^\#_s+Z^P_s)\big)s^b-s^\epsilon\Bt^\epsilon_s+hY^\#_s\bigg]\df s\nonumber\\
          &-\int_{t}^{T} h^HL^HL^m \pt_s\df s
            -\int_{t}^{T}(Z^\#_s)^\top\df W_s, \nonumber 
\end{align}
Then we will show that:
\begin{align}
 Y^\# +L^m\pt +\Bt^\epsilon- \alpha(Z^\#+Z^P) \leq 0,~~\dqdt. \label{fca-dir}
\end{align}
To demonstrate this, we take another transformation,
$V^\dsF \coloneqq Y^\# + L^m\pt+\Bt^\epsilon$ and $\Pi^\dsF \coloneqq Z^\# +
L^mZ^P$. Then, $(V^\dsF, \Pi^\dsF)$ is the solution of:
\begin{align}
  V^\dsF_t =& L^m\xi + \Bt^\epsilon_T + \int_{t}^{T}F_t(V^\dsF_s, \Pi^\dsF_s)\df s  
      -\int_{t}^{T}(\Pi^\dsF_s)^\top\df W_s,\label{bsdess}
\end{align}
where $\xi \coloneqq \St^1_T -B^{-1}_TK$ and:
\begin{align}
    &F_t(y, z)\coloneqq-(y -\alpha z)s^{b}-hy+\xi^b_t,\label{ex.idea.F}\\
  &\xi^b \coloneqq
    (h-h^HL^HL^m )\pt + h\Bt^\epsilon + \alpha(1-L^m)s^bZ^P.\label{ex.idea.xib}\nonumber 
\end{align}
Under mild conditions, we can obtain
$(V^\dsF, \Pi^\dsF)\in \dsL^2([0, T]\colon \dsD^{1, 2} \times \dsD^{1, 2})$, i.e., for
any $t\leq T $, $(V^\dsF_t, \Pi^\dsF_t) \in \dsD^{1, 2} \times \dsD^{1, 2}$ and:
\begin{align}
  \int_{0}^{T} \big(\Vert V^\dsF_t \Vert^2_{1, 2} + \Vert \Pi^\dsF_t \Vert^2_{1, 2}\big)\df t < \infty. 
\end{align}
Furthermore, $(D_tV^\dsF_t)_{0\leq t \leq T}$ is a version of
$(\Pi^\dsF_t)_{0\leq t \leq T}$. In addition, note that (\ref{fca-dir}) is equivalent to:
\begin{align}
  V^\dsF  -(1-L^m)s^b\alpha Z^P- \alpha \Pi^\dsF \leq V^\dsF -(1-L^m)s^b\alpha Z^P - \alpha DV^\dsF\leq0, ~~\dqdt. \label{fca-dir2}
\end{align}
To show (\ref{fca-dir2}), 
 for $\theta \leq t$, let $(V^\dsF_{t, \theta}, \Pi^\dsF_{t, \theta})\coloneqq(\alpha D_\theta V^\dsF_t,
 \alpha D_\theta \Pi^\dsF_t)$. Therefore, $(V^\dsF_{\cdot, \theta}, \Pi^\dsF_{\cdot, \theta})$ is given by:
\begin{align}
  &V^\dsF_{t, \theta}
  =\alpha L^mD_\theta\xi + \int_{t}^{T}F_{s, \theta}( V^\dsF_{t, \theta},  \Pi^\dsF_{t, \theta})\df s
                     -\int_{t}^{T}(\Pi^\dsF_{s, \theta})^\top\df W_s, \label{bsdessm}\\
  &F_{t, \theta}(y, z)\coloneqq-(y -\alpha z)s_t^{b}-h_ty +\alpha D_{\theta}\xi^b_t.
\end{align}
Note that $F_{t, \theta}(y, z) = F_{t}(y, z) +\alpha D_{\theta}\xi^b_t-\xi^b_t $. Specifically,
(\ref{bsdessm}) can be written as:
\begin{align}
  V^\dsF_{t, \theta}
  =&\alpha L^mD_\theta\xi + \int_{t}^{T}\Big[F_{s}
     ( V^\dsF_{s, \theta},  \Pi^\dsF_{s, \theta})+\alpha
     D_{\theta}\xi^b_s-\xi^b_s\big]\df s 
                     -\int_{t}^{T}(\Pi^\dsF_{s, \theta})^\top\df W_s. \label{bsdessm2}  
\end{align}
We will show  $V^\dsF_{\cdot} \leq V^\dsF_{\cdot , \theta} $ by comparing (\ref{bsdess}) and
(\ref{bsdessm}), and it suffices to show that for $\theta \leq
t$:
\begin{align}
  L^m\xi+\Bt^\epsilon_T-\alpha L^mD_{\theta}\xi \leq&0, ~~\text{a.s, }\label{ex.idea.ineq1}\\
  \xi^b -\alpha D_{\theta}\xi^b \leq& 0, ~~ \dqdt,\label{ex.idea.ineq2}\\
  Z^P \geq& 0, ~~ \dqdt.
\end{align}
We will show that the above inequalities hold if $\epsilon$ is not too large. More
precisely, we assume:
\begin{align}
  \epsilon \leq \epsilon_*\coloneqq K (B^\ell_T)^{-1}\min\Big\{L^m , ~ 1-\frac{h^HL^HL^m}{h}\Big\}. \label{ex.idea.eus}
\end{align}
It is easy to verify:
\begin{align}
  \Bt^\epsilon_T+L^m\xi-\alpha L^m D_{\theta}\xi =
  &\Bt^\epsilon_T+ L^m(\St^1_T - B^{-1}_TK)-(\sigma^1)^{-1}L^m D_{\theta}(\St^1_T - B_T^{-1}K)
                                   \nonumber \\
  =&\Bt^\epsilon_T+L^m(\St^1_T - B_T^{-1}K) - L^m\St^1_T=\Bt^\epsilon_T- L^mB_T^{-1}K\leq0. \nonumber 
\end{align}
Moreover, by Proposition 3.12 in \cite{di2009malliavin}:
\begin{align}
  \pt_t -\alpha D_{\theta}\pt_t
  =&\dsE[\xi|\calF_t]-\alpha D_{\theta}\dsE[\xi|\calF_t]
     =\dsE[\xi|\calF_t]-\alpha\dsE[ D_{\theta}\xi|\calF_t] =\dsE[\xi-\alpha D_{\theta}\xi|\calF_t].\nonumber 
\end{align}
Furthermore, $Z^P = D\pt = (\sigma^1)^{-1}\St^1$. It then follows that:
\begin{align} 
  \xi^b_t -\alpha D_{\theta}\xi^b_t =&
  (h-h^HL^HL^m)  \dsE[\xi-\alpha D_{\theta}\xi|\calF_t] +h\Bt^\epsilon_t \nonumber \\
  =&h(\Bt^\epsilon_t - B_T^{-1}K) + h^HL^HL^mB_T^{-1}K.
\end{align}
Therefore, by the comparison principle, for $\theta \leq t$, $V^\dsF_{t} \leq V^\dsF_{t ,
  \theta}$. In particular, $V^\dsF_{t} \leq V^\dsF_{t ,  t} = \alpha D_t V^\dsF_{t} = \alpha
\Pi^\dsF_{t}$. Consequently, (\ref{fca-dir}) is guaranteed and  $(Y^{\#}, Z^{\#}) =
(Y^\dsF, Z^\dsF)$. In addition, $(Y^\dsF, Z^\dsF)$ follows the linear BSDE
(\ref{bsdess}), so we can find an analytical form of $(Y^\dsF, Z^\dsF)$, i.e.,
$(V, \pit)$ as well. Moreover,  $\text{FBA} = 0$.
\qed
\end{example}
\begin{remark} \label{idea:remark}
\begin{enumerate}[(i)]
\item If $\epsilon \geq \epsilon^*\coloneqq K (B^b_T)^{-1}$, one may want to consider $(Y^\#,
  Z^\#)$ given by:
  \begin{align}
     Y^\#_t =& -\int_{t}^{T} \bigg[\big(Y^\#_s +L^m\pt_s +\Bt^\epsilon_s
            -\alpha(Z^\#_s+Z^P_s)\big)s^\ell-s^\epsilon\Bt^\epsilon_s+hY^\#_s\bigg]\df s\nonumber\\
          &-\int_{t}^{T} h^HL^HL^m \pt_s\df s
            -\int_{t}^{T}(Z^\#_s)^\top\df W_s, \nonumber      
  \end{align}
and obtain opposite inequalities of (\ref{ex.idea.ineq1}) and
(\ref{ex.idea.ineq2}). In this case, $\text{FCA}=0$. Therefore, the binary funding
impacts depend on the value of initial portfolio, $\epsilon$. 
\item Consider the same market conditions but
 $\frD = \1_{\llbracket T, \infty \llbracket }(K-S^1_T)$. 
In other words, the hedger is in a long position of the stock forward contract. We can
calculate  the opposite inequality of (\ref{fca-dir}) in the same manner. In this
case, $\text{FCA}=0$, but $\text{FBA}\not=0$. Therefore, offset of a
substantial portion between FCA and FBA is hardly expected.
\end{enumerate}
\end{remark}
\section{Stochastic Intensities}
\label{app:sto.int}
We assume that $h^i$, $i \in \{H, C\}$, are $\dsF$-adapted processes, but
$\sigma^H$ and $\sigma^C$ are deterministic. For simplicity, we set
$\epsilon = 0$. Let $n=2$, $\rho =\{H, C\}$, $B_t^{-1}\frD_t = \1_{T \leq t}\xi$ for some
$\xi\geq0$, and we consider \textit{replacement cost}. $\xi$ is determined by an
asset $S^1$, but the market is completed by another non-defaultable traded asset
$S^2$. As in the main sections, we assume that
 $ \Sigma \coloneqq [\sigma^1 ~~\sigma^2]^\top $
is of full rank.  In this case, we cannot expect that the transformation,
 $\phi_t   \colon  \pit_t^\top\sigma_t-Z^P \to  \sum_{i \in I\setminus \rho} \pit^i_t$,
is independent of $\pit^H$ and $\pit^C$. By \cref{fsol3},
$\pit^i$, $i \in \{H, C\}$, are related to  $Y^\dsF$. Therefore, we write $\phi$ as
$  \phi_t(y, z)$. 
 Then:
\begin{align}
  \phi_t(y, z)
  =& \alpha_t^\top\big(z+Z^P_t-(\sigma^H_T)^\top\pit^H_t-(\sigma^C_T)^\top\pit^C_t\big)\nonumber \\
  =& \1^\top(\Sigma^\top_t)^{-1}\big[z+Z^P_t-(\sigma^H_t)^\top L^HL^m(y + \pt_t)\big].\nonumber
\end{align}
Let us consider:
\begin{align}
  Y^\#_t =& \int_{t}^{T} -\bigg[\Big(L^mY^\#_s +L^m\pt_s
            +\sigma^H_s\alpha_s L^HL^m(Y^\#_s+\pt_s)-\alpha_s^\top(Z^\#_s+Z^P_s))\Big)s^{\ell}_s
            \bigg]\df s\nonumber\\
  &+\int_{t}^{T}
    \big[-h^H_sL^HL^m(Y^\#_s+\pt_s)\big]\df s
    -\int_{t}^{T}Z^\#_s\df W_s,  \nonumber 
\end{align}
We will show that $(Y^\#, Z^\#) = (Y^\dsF, Z^\dsF)$. To this end, let $V^\dsF
\coloneqq Y^\#+\pt$, $\Pi^\dsF \coloneqq Z^\#+Z^P$, and $\beta \coloneqq
L^m(1+\sigma^H\alpha L^H)$. Then $(V^\dsF, \Pi^\dsF)$ is given by the following:
\begin{align} 
  &V^\dsF_t =\xi +\int_{t}^{T} F_s(V^\dsF_s, \Pi^\dsF_s)\df s
                 -\int_{t}^{T}\Pi^\dsF_s\df W_s,\label{appen:bsde}\\
  &F_t(y, z) \coloneqq-(\beta_t y - \alpha^\top_t z\big)s^\ell_t
                 - h^H_tL^HL^my.
\end{align}
It suffices to show:
\begin{align}
  \beta V^\dsF - \alpha^\top \Pi^\dsF \geq0, ~~~\dqdt. \label{appen:ineq}
\end{align}  
We denote that for $\theta \leq t$:
\begin{eqnarray}
  &V^\beta_{t, \theta} \coloneqq \beta_\theta V^\dsF_t, ~~ \Pi^\beta_{t, \theta} \coloneqq \beta_\theta
    \Pi^\dsF_t,\nonumber \\
  &V^D_{t, \theta} \coloneqq \alpha_\theta^\top D_\theta V^\dsF_t, ~~ \Pi^D_{t, \theta} \coloneqq \alpha^\top_\theta D_\theta
    \Pi^\dsF_t . \nonumber 
\end{eqnarray}
Then, $(V^\beta_{t, \theta}, \Pi^\beta_{t, \theta})$ and $(V^D_{t, \theta}, \Pi^D_{t, \theta})$ are given by the following:
\begin{align}
 V^\beta_{t, \theta} =&\beta_\theta\xi +\int_{t}^{T} F_s(V^\beta_{s, \theta}, \Pi^\beta_{s, \theta})\df s
               -\int_{t}^{T}\Pi^\beta_{s, \theta}\df W_s,\nonumber \\
  V^D_{t, \theta} =&\alpha_\theta^\top D_\theta\xi +\int_{t}^{T} \big[F_s(V^D_{s, \theta}, \Pi^D_{s, \theta})
                -\alpha_\theta^\top (D_\theta h_s)L^HL^m V^\dsF_s\big]\df s
                -\int_{t}^{T}\Pi^D_{s, \theta}\df W_s, \nonumber 
\end{align}
Recall $\xi \geq0$. Thus, by (\ref{appen:bsde}), $V^\dsF \geq0$. Consequently, if
$\beta_\theta\xi \geq \alpha_\theta^\top D_\theta\xi$ and $D_\theta h^H_t \geq 0$, \cref{appen:ineq} is satisfied. Therefore,
FCA$=0$.  
\end{document}